\documentclass[11pt]{article}
\usepackage[top=1.0in, bottom=1.0in, left=1in, right=1in]{geometry} 
\usepackage[T1]{fontenc}
\usepackage{textcomp}
\usepackage{dsfont}		
\usepackage{comment}
\usepackage[frenchmath]{mathastext}  			
\usepackage{algorithm}
\usepackage{algpseudocode}
\usepackage{amsmath,amssymb}
\usepackage{physics}
\usepackage{bm}
\usepackage{bbm}
\usepackage{graphicx} 										
\usepackage[font={small}]{caption}
\usepackage{booktabs}
\usepackage{multirow}
\usepackage{subfigure}
\usepackage{float}
\usepackage{upgreek}
\usepackage{enumitem}
\usepackage{mdwlist}												
\usepackage[dvipsnames]{xcolor}							
\usepackage[style=numeric, natbib=true,uniquename=true, maxcitenames=2, maxbibnames=999, uniquelist=false, backend=biber]{biblatex}
\usepackage{todonotes}

\usepackage[plainpages=false, pdfpagelabels]{hyperref} 
	\hypersetup{
		colorlinks   = true,
		citecolor    = RoyalBlue,
		linkcolor    = RubineRed, 
		urlcolor     = RubineRed
	}
\usepackage{xpatch}

\usepackage[mathscr]{eucal}                 
\usepackage{ulem}
\usepackage{xcolor,mathtools}                      
\mathtoolsset{showonlyrefs=true}

\usepackage{amsthm}
\usepackage{thmtools}            
	\allowdisplaybreaks                   
	\theoremstyle{plain}
	\newtheorem{theorem}{Theorem}
	\newtheorem{lemma}[theorem]{Lemma}       
	\newtheorem{proposition}[theorem]{Proposition}
	\newtheorem{corollary}[theorem]{Corollary}
	\theoremstyle{definition}

	\newtheorem{remark}{Remark}

\usepackage{soul}
\usepackage{authblk}
\usepackage{setspace,lipsum}

\newcommand\Eb{\mathds{E}}

\newcommand\Hb{\mathds{H}}
\newcommand\Ib{\mathds{1}}
\newcommand\Lb{\mathds{L}}
\newcommand\Pb{\mathds{P}}

\newcommand\Rb{\mathds{R}}
\newcommand\Sb{\mathds{S}}

\newcommand\Fc{\mathcal{F}}

\newcommand\Lc{\mathscr{L}}

\newcommand\Yc{\mathcal{Y}}
\newcommand\Zc{\mathcal{Z}}

\newcommand\eps{\varepsilon}

\newcommand\Keps{K^\varepsilon}

\newcommand\Yeps{Y^\varepsilon}
\newcommand\Zeps{Z^\varepsilon}
\newcommand\feps{f^\varepsilon}
\newcommand\phieps{\phi^\varepsilon}

\newcommand\Ycb{\overline{\Yc}}
\newcommand\YyZc{Y^{y,\bm{\Zc}}}
\newcommand\YcZc{\Yc^{\bm{\Zc}}}
\newcommand\YcbZc{\overline{\Yc}^{\bm{\Zc}}}

\newcommand\Zct{\widetilde{\Zc}}

\newcommand{\Var}{\mathrm{Var}}
\renewcommand{\d}{\partial}

\newcommand{\ess}{\mathrm{ess}}

\newcommand{\eqlnostar}[2]{\begin{align}\label{#1}#2\end{align}}
\newcommand{\eqstar}[1]{\begin{align*}#1\end{align*}}
\newcommand{\eq}[1]{\ifthenelse{\equal{#1}{*}}
  {\eqstar}
  {\eqlnostar{#1}}
 }
\newlist{describe}{description}{1}
\setlist[describe,1]{%
  font=\normalfont\textsf,
  itemindent=0pt,
  wide,
  itemsep=0pt,topsep=2pt,
  format={\normalfont\textsfcolor}
}
\newcommand*{\textsfcolor}[1]{\textsf{#1:}}
\makeatletter
\xpatchcmd{\enit@description@i}{%
  \labelsep\z@
}{%
  \phantomsection
  \let\org@label\label
  \let\label\@gobble
  \protected@edef\@currentlabel{##1}%
  \let\label\org@label
  \labelsep\z@
}{}{\undefined}
\makeatother

\begin{document}

\title{Deep Learning for Reflected BSDEs: Regularization and Error Analysis}

\author{Ruimeng Hu\thanks{Department of Mathematics, and Department of Statistics and Applied Probability,
University of California, Santa Barbara, CA 93106-3080, USA.
\, Email: \texttt{rhu@ucsb.edu}.}
\qquad 
Yihan Zou\thanks{Adam Smith Business School, University of Glasgow, Glasgow, G12 8QQ, United Kingdom.\, Email: \texttt{yihan.zou@outlook.com}.}
}

\date{\today}

\maketitle

\begin{abstract}

Reflected backward stochastic differential equations (RBSDEs) provide a probabilistic formulation for obstacle constrained problems, but existing deep learning methods for their high dimensional solution remain limited. In this paper, we propose two deep learning schemes for RBSDEs, a deep forward scheme (DFS) and a deep backward scheme (DBS), by first reducing the reflected problem to a family of regularized BSDEs. Our main theoretical contribution concerns the DBS: we establish an explicit error bound showing that, for each fixed regularization parameter $\eps>0$, the approximation error between the DBS solution and the solution to the regularized BSDE is controlled by the associated training loss. We prove that this training loss can be controlled by the universal approximation capability of neural networks. Together, these results yield a theoretical foundation for the deep learning-based solution and complement existing analysis for forward type methods. We illustrate the framework on high dimensional American option pricing, where the reflected formulation allows us to address the continuous time exercise feature directly rather than through a Bermudan approximation. Numerical experiments demonstrate that both DFS and DBS deliver accurate solutions in high dimensions.

\smallskip


\end{abstract}


\tableofcontents


\section{Introduction}

Reflected backward stochastic differential equations (RBSDEs) arise naturally in stochastic control, optimal stopping, variational inequalities, and mathematical finance. A prominent example is the pricing of American-style options, where the reflection enforces the early-exercise constraint. Despite their wide range of applications, the numerical solution of RBSDEs remains challenging, especially in high dimensions. In this paper, we develop and analyze deep learning-based numerical methods for solving RBSDEs in high dimensions.

Let $(\Omega,\Fc,\{\Fc_t\}_{0\le t\le T},\Pb)$ be a filtered probability space, where $\Fc = \Fc_T$ and $\{\Fc_t\}_{0\le t\le T}$ is the natural filtration generated by a $d$-dimensional ($d\ge 1$) Brownian motion $W$, augmented by the $\Pb$-null sets (the usual conditions of right-continuity and completeness are satisfied). We consider the forward diffusion
\eqlnostar{eq-SDE}{X_t = x + \int_0^t b(s,X_s)\dd s + \int_0^t \sigma(s,X_s)\dd W_s,}
with initial condition $x\in\Rb^d$, together with  the reflected backward equation
\eqlnostar{eq-RBSDE}{\left\{\begin{array}{l}
Y_t = \Phi(X_T) + \int_t^T f\left(s, X_s, Y_s, Z_s\right)\dd s + K_T - K_t - \int_t^T Z_s \dd W_s, \\
Y_t \geq S_t, \\
\int_0^T\left(Y_t-S_t\right) \dd K_t=0,\end{array}\right.}
where $S_t = \Phi(X_t)$ and the process $K$ is continuous and non-decreasing with $K_0 = 0$. Under standard Lipschitz and integrability assumptions on the coefficients, existence and uniqueness of a solution $(Y,Z,K)$ are well known; see \citet{el1997reflected}. Our objective is to approximate the pair $(Y,Z)$ efficiently and accurately in high dimensions. 

Compared with standard BSDEs, the main difficulty in an RBSDE lies in the reflection process $K$. The presence of the reflection process $K$ constitutes the main difficulty to numerical approximation. Unlike standard BSDEs, where the solution $(Y,Z)$ can often be characterized through conditional expectations, the process $K$ is defined only implicitly through the Skorokhod condition and depends on the entire trajectory of $Y$. This feature makes both the design and the analysis of numerical methods substantially more delicate. Classical discretization and regression methods for BSDEs, such as those in \citet{bouchard2004discrete,lemor2006rate}, therefore do not extend directly to the reflected setting. 

Classical numerical methods for RBSDEs include time discretization combined with regression techniques \cite{ma2005representations,bouchard2008discrete}, penalization schemes \cite{gobet2026improved}, etc. Regression-based schemes rely on least-squares approximations of conditional expectations, while the penalization method approximates the reflection by adding increasingly stiff terms to the driver. Although these methods are effective in low or moderate dimensions, their performance deteriorates in high dimensions. In particular, least-squares regression, which relies on approximating conditional expectations on finite-dimensional function spaces, becomes increasingly expensive and inaccurate as the dimension grows, because the complexity of the approximation space increases rapidly.

More recently, deep learning has emerged as an effective tool for solving high-dimensional BSDEs, motivated by the universal approximation ability of neural networks. Several deep learning schemes for BSDEs reformulate the problem as a (sequence of) global optimization problem over network parameters and have been shown to have strong empirical and theoretical performance for standard BSDEs in high dimensions (see, e.g., \cite{han2018solving,han2020convergence,hure2020deep,germain2022approximation,wang2018deep,andersson2023convergence,chen2021deep,gao2023convergence,andersson2025deep,gnoatto2025convergence,ji2022solving,negyesi2025deep,huang2026compounded}).

For reflected BSDEs, however, the presence of $K$  still poses a major challenge. On the one hand, several convergence results for deep backward schemes have been developed only for standard, non-reflected BSDEs. For instance, \citet{gao2023convergence} establish convergence of a backward deep BSDE method under assumptions tailored to the non-reflected setting, including small Lipschitz constants of the driver with respect to the solution variables. Such arguments do not transfer directly to RBSDEs because of the additional reflection constraint. On the other hand, there are deep-learning methods designed more specifically for reflected or obstacle-type problems. A notable example is the reflected deep backward dynamic programming method (RDBDP) of \citet{hure2020deep}, where the obstacle is enforced by taking the maximum between the obstacle and a continuation value approximated by neural networks. This method solves a sequence of local optimization problems, one at each time step. As pointed out in the subsequent work \citet{germain2022approximation}, maintaining stability in this backward procedure may require many stochastic gradient iterations at early time steps, so the computational cost can increase rapidly as the time discretization becomes finer. More recently, \citet{bayraktar2024deep} extend the RDBDP framework to path-dependent American option pricing. We also mention related neural-network-based methods for optimal stopping problems in \cite{becker2019deep,becker2020pricing,hu2020deep,kohler2010pricing,lapeyre2021neural,reppen2023deep,reppen2025neural,herrera2024optimal}. There is also a related reinforcement learning literature on optimal stopping, in which stopping decisions are either randomized or treated through penalized Hamilton-Jacobi-Bellman (HJB) formulations. For example, \citet{dong2024randomized} studies an intensity-based randomized stopping framework with an unnormalized negentropy regularizer, while \citet{dai2026learning} consider a penalized HJB with Bernoulli randomized controls and the usual Shannon entropy regularization.

A different route has recently been proposed in \citet*{agarwal2026numerical}, where the RBSDE is approximated by a family of regularized BSDEs. In this framework, the reflection constraint is replaced by a smooth regularization term added to the drive of $Y$, and the solution of the regularized BSDE converges to that of the original RBSDE as the regularization parameter $\eps\downarrow 0$. This regularization removes the explicit reflection process from the equation and yields a formulation that is considerably more amenable to deep learning solvers. In particular, it provides a natural bridge from reflected BSDEs to standard BSDEs that learning methods can be applied.

Motivated by this regularization framework, we propose two deep learning-based numerical schemes for RBSDEs: a deep forward scheme (DFS) and a deep backward scheme (DBS). Both methods are applied to the regularized BSDE approximation of the reflected problem. The DFS follows the deep BSDE philosophy \cite{han2018solving,han2020convergence}, while the DBS shares the spirit of \cite{wang2018deep},  both allowing the backward components to be approximated through neural networks trained via stochastic optimization. 

Our main theoretical contribution concerns the DBS. More specifically, under standard regularity assumptions, we derive an explicit error estimate showing that the approximation error of the DBS is controlled by its training loss function. We then prove that this loss can be made small under the universal approximation property of neural networks. Together, these results yield an error analysis for the proposed deep backward solver in the reflected setting. In contrast to several existing convergence results for deep BSDE methods, our analysis does not rely on smallness assumptions on the Lipschitz coefficient of the driver with respect to the solution variables.

From the computational perspective, the proposed framework also allows us to address continuous-time American option pricing through the reflected BSDE formulation itself, rather than only through Bermudan approximations with finitely many exercise dates (the case in most existing literature). This feature is particularly appealing in high dimensions, where accurate treatment of the early-exercise constraint is a major challenge. Our numerical experiments show that both DFS and DBS achieve high accuracy on high-dimensional American-style option pricing problems. To the best of our knowledge, many existing deep learning methods for high-dimensional early-exercise options typically rely on time discretization and therefore effectively approximate Bermudan problems, whereas our approach directly targets the reflected formulation associated with the continuous-time American problem.

The remainder of the paper is organized as follows. Section \ref{sec:DBS} introduces the model setup, standing assumptions, and the proposed deep learning solver for RBSDEs. Section \ref{sec:convergence analysis} is devoted to the error analysis of the deep backward scheme. In Section \ref{sec:numerical results}, we illustrate the accuracy of the proposed deep learning algorithms through numerical experiments for high-dimensional American-style option pricing.

\section{Deep Learning Schemes for RBSDEs via Regularization}\label{sec:DBS}
We start this section by introducing some preliminary notations and assumptions that will be used later. Then, we present a regularized BSDE that approximates the RBSDE, together with deep learning-based schemes for the regularized BSDE to solve the RBSDE numerically.

\subsection{Notation, Function Spaces, and Assumptions}
For any $p \ge 2$, we use the following notation. For $x\in\Rb^d$, $x^\top$ denotes its transpose and $|x|$ its Euclidean norm. For a scalar-valued function $g=g(t,x,y,z)$ with $x=(x_1,\dots,x_d)\in\Rb^d$, we write  $\d_x g:=\big(\frac{\d g}{\d x_1},\ldots,\frac{\d g}{\d x_d}\big)$ for the gradient with respect to (w.r.t.) $x$, and similarly for $\partial_y g$ and $\partial_z g$. We denote expectation under $\Pb$ by $\Eb[\cdot]:=\Eb[\cdot|\Fc_0]$, and the conditional expectation w.r.t. $\mathcal{F}_{t_i}$ by $\Eb_{t_i}[\cdot]:=\Eb[\cdot|\Fc_{t_i}]$, where $0=t_0 < t_1 \cdots < t_N=T$ is a given time partition on $[0,T]$. For a stochastic process $\psi=\{\psi_t, 0\le t\le T\}$, let $\psi^*_t := \underset{0 \le s \le t}{\sup} \left|\psi_s\right|$ be the running maximum of its absolute value. We further introduce the spaces $\Lb^p(\Fc)$, $\Sb^p$, $\Hb^p$ and $C^{\infty}(\Rb^d)$ defined by:
\begin{itemize}
    \item $\Lb^p(\Fc) := \left\{\Fc\text{-measurable random variables}\; \zeta:\|\zeta\|^p_p := \Eb\left[|\zeta|^p\right]<\infty\right\}$.

    \item $\Sb^p := \Big\{\text{adapted processes}\;\psi = \{\psi_t,\;0\le t\le T\}: \|\psi\|_{\infty,p}^p := \Eb\big[\underset{0 \le t \le T}{\sup}\left|\psi_t\right|^p\big] < \infty\Big\}$.

    \item $\Hb^p := \Big\{\text{predictable processes}\;\psi = \{\psi_t,\;0\le t\le T\}: \|\psi\|_{p}^p := \Eb\big[\big(\int_0^T\left|\psi_t\right|^2\dd t\big)^{p/2}\big] < \infty\Big\}$. 

    \item $C^{\infty}(\Rb^d) :=$ $\Big\{$infinitely continuously differentiable functions $g:\Rb^d\rightarrow\Rb$ $\Big\}$.  
\end{itemize}

We impose the following standing assumptions throughout the paper.
\begin{description}
\item[Assumption ($\textbf{H}_1$)\label{hp-1}]
The coefficients of the forward SDE \eqref{eq-SDE} and the RBSDE \eqref{eq-RBSDE} satisfy 
    \begin{itemize}
        \item $b(\cdot,0)$, $\sigma(\cdot,0)$, are $f(\cdot,0,0,0)$ are bounded, and $\Phi(0)$ is finite.
        
        \item The coefficients $b,\sigma,\Phi$ and the generator $f$ are uniformly $C_f$-Lipschitz continuous in their space variables, i.e. for all $t\in[0,T]$,
        \eqlnostar{}{
        & |b(t,x)-b(t,x')| + |\sigma(t,x)-\sigma(t,x')| + |\Phi(x)-\Phi(x')| \le C_f|x-x'|,\\
        & |f(t,x,y,z) - f(t,x',y',z')| \le C_f\left(|x - x'| + |y - y'| + |z - z'|\right),
        }
        for all $y,y'\in\Rb, x,x',z,z'\in\Rb^d$.
        
        \item The functions $b, \sigma, f$ are uniformly $\tfrac{1}{2}$-H\"older continuous in $t$.
    \end{itemize}
\end{description}

\begin{description}
\item[Assumption ($\textbf{H}_2$)\label{hp-2}]
The obstacle $S := \{S_t, 0\le t\le T\}$ is a progressively measurable scalar process s.t. $S\in\Sb^p$. Furthermore, $S$ satisfies:
    \begin{itemize}
        \item The obstacle $S$ admits the following semi-martingale representation:
        \eqlnostar{eq-representation of S}{S_t = S_0 + \int_0^t V_s\dd W_s + \int_0^t U_s\dd s + A_t,}
        where $V\in \Sb^p$, $U\in \Sb^p$, $A$ is a continuous, non-decreasing process of finite variation with
        $A_0 = 0$ and
        $\int_0^T |U_t|\dd t + A_T <\infty$.
        We further assume that $\dd A_t$ is singular w.r.t. the Lebesgue measure, so that the above decomposition is unique.
        \item $\overline{\kappa}_{\infty} := \underset{(t,\omega)\in[0,T]\times\Omega}{\ess\;\sup}\;\kappa_t(\omega)< \infty$ where \begin{align}\label{eq:kappa}\kappa_t := \left(f(t,X_t,S_t,V_t) + U_t\right)^-, \end{align}
        and $f$ is the generator of the RBSDE \eqref{eq-RBSDE}.
        \item For any process $\psi \in \{S,U,V\}$, we have $\Eb|\psi_t - \psi_s|^2 \le C |t-s|$, for $0\le s\le t\le T$.
    \end{itemize}
\end{description}

\subsection{A Regularized Approximation of the RBSDE}\label{section:regularized RBSDE}
We next introduce a smooth approximation of the step function $\Ib_{(-\infty,0]}$, which will be used to construct a regularized version of the RBSDE. For $\eps>0$, define the mollifier $\phieps:\Rb\to[0,1]$ by
\eqlnostar{eq:phi^epsilion}{\phieps(x) := \phi\left(\frac{x}{\eps}\right),}
where $\phi: \Rb\rightarrow [0,1]$ is  a smooth non-increasing function (i.e. $\phi \in C^{\infty}(\Rb)$) satisfying $\phi(x)=1$ for $x\le 0$ and $\phi(x)=0$ for $x\ge 1$. A possible choice of $\phi$ is
\eqlnostar{}{\phi(x) :=
\begin{cases}
1, & \text{if } x \le 0, \\
\dfrac{
  \exp\!\bigl(-\tfrac{1}{1 - x}\bigr)
}{
  \exp\!\bigl(-\tfrac{1}{x}\bigr)
  \;+\;
  \exp\!\bigl(-\tfrac{1}{1 - x}\bigr)
}, & \text{if } 0 < x < 1, \\
0, & \text{if } x \ge 1.
\end{cases}}
Since $\phi \in C^\infty(\Rb)$, without loss of generality, we assume that the first derivative of $\phi$ is bounded by a constant $C_\phi$, i.e. $\sup_{x \in \Rb} |\d_x\phi| \leq C_\phi$. Then, we have the following regularity property of $\phieps$.
\begin{lemma}\label{lemma:regularity of phi^epsilon}
For every $\forall h>0$, we have that $\left|\phieps(x+h) - \phieps(x)\right| \le C_{\phi}\eps^{-1} h$.
\end{lemma}

Given a regularization parameter $\eps>0$, the regularized RBSDE (ReRBSDE in short) associated with \eqref{eq-RBSDE} is given by
\eqlnostar{eq-ReRBSDE}{\left\{\begin{array}{l}\Yeps_t = \Phi(X_T) + \int_t^T \feps\left(s, X_s, \Yeps_s, \Zeps_s\right)\dd s - \int_t^T \Zeps_s \dd W_s, \quad \forall 0 \le t \le T, \\
\Keps_t := \int_0^t \phieps(\Yeps_s-S_s)\kappa_s\dd s.\end{array}\right.}
where $\kappa_t$ is defined in \eqref{eq:kappa}, and the regularized generator $\feps$ is given by
\eqlnostar{eq-regularized generator}{\feps(t,x,y,z) := f(t,x,y,z) + \phieps(y-\Phi(x))\kappa_t,} 
with $S_t = \Phi(X_t)$ and $\phieps$ defined in \eqref{eq:phi^epsilion}.

The regularization removes the singular reflection term $K$ from the backward equation and replaces it with a smooth penalty term $\phieps(\cdot)\kappa_t$. As $\eps\downarrow 0$, the solution of the ReRBSDE  \eqref{eq-ReRBSDE} converges to that of the original RBSDE \eqref{eq-RBSDE}. The next theorem summarizes the approximation result established in \citet{agarwal2026numerical}.

\begin{theorem}[{\citet[Theorems 6 and 7]{agarwal2026numerical}}]\label{thm:a priori}
Assume \ref{hp-1}-\ref{hp-2}. Let $p\ge 2$, $(Y,Z, K)\in\Sb^p\times\Hb^p\times \Sb^p$ be the solution of the RBSDE \eqref{eq-RBSDE}, and $(\Yeps,\Zeps, \Keps)\in\Sb^p\times\Hb^p \times \Sb^p$ be the solution of the ReRBSDE \eqref{eq-ReRBSDE}. Then the triplet $(\Yeps,\Zeps,\Keps)$ is bounded in $\Sb^p\times\Hb^p\times\Sb^p$ uniformly in $\varepsilon$. Moreover,
\eqlnostar{eq: regularization error}{\Eb\Big[\left|\left(Y - \Yeps\right)^*_T\right|^p + \Big(\int_0^T \left|Z_t - \Zeps_t\right|^2\dd t\Big)^{p/2} + \left|\left(K - \Keps\right)^*_T\right|^p\Big]\le C {\varepsilon}^{p/2}\Eb\left[\left|\Keps_T\right|^{p/2}\right]\le C {\varepsilon}^{p/2},}
where $C$ is a generic positive constant depending only on $\Phi$, $f$, $S$, $T$, and  $p$ (whose value may change from line to line). Additionally, if $f$ is non-increasing in $y$, we have
\eqlnostar{}{0\le \Yeps_t - Y_t\le \varepsilon,\;\forall t\in[0,T],\; \Pb-\text{a.s.}}
\end{theorem}

\begin{remark}
The above theorem shows that the regularized solution $(\Yeps,\Zeps)$ converges to the solution $(Y,Z)$ of the original RBSDE as the regularization parameter $\eps$ vanishes. Separately, in the sequel we will show that the distance between the regularized solution $(\Yeps,\Zeps)$ and the deep learning-based solver $(\YcZc,\Zc)$, introduced in Section \ref{sec:deep forward and backward schemes}, can be controlled by the associated loss functional, which can be made arbitrarily small subject to universal approximation. Combining these results, we conclude that the solver $(\YcZc,\Zc)$ can yield an arbitrarily accurate approximation of the solution $(Y,Z)$ to the RBSDE. This conclusion is formalized in \autoref{thm:true DBS approx error} and \autoref{thm:controlofvar}.
\end{remark}

\subsection{Deep Forward and Backward Schemes}\label{sec:deep forward and backward schemes}

Let $0=t_0<t_1<\cdots<t_N=T$ be an equidistant time partition with step size $h = \tfrac{T}{N}$, and define the Brownian increments by $\Delta W_i = W_{t_{i+1}} - W_{t_{i}}, \;\forall i=0,\ldots,N-1$. We next propose two deep learning schemes for approximating the solution of the regularized RBSDE \eqref{eq-ReRBSDE}: a deep forward scheme (DFS), inspired by \citet{han2018solving}, and a deep backward scheme (DBS), inspired by \citet{wang2018deep}.

For both schemes, we use a family of neural networks 
\[
\bm{\Zc} := (\Zc_i)_{i=0}^{N-1}, \qquad \Zc_i(\cdot; \theta):\Rb^d \rightarrow \Rb^d, \quad i=0,\ldots,N-1,
\]
to denote the sequence of networks that approximate $\Zeps$ on the time grid. We write $\theta$ for the collection of all trainable network parameters. In the DFS, we also introduce a trainable scalar parameter $y$ to parameterize the initial value $Y_0$.

To focus on the approximation and discretization errors of the backward equation, we assume throughout that the forward process $X$ can be sampled exactly on the grid. This is the case, for example, for models with explicit transitions such as geometric Brownian motion, Ornstein-Uhlenbeck processes, and Cox-Ingersoll-Ross processes. When exact simulation is not available, one may instead use an Euler-type discretization of $X$, which introduces the usual strong error of order $h^{1/2}$ and weak error of order $h$. We now describe the two algorithms. 

\medskip
\noindent
\textbf{Deep forward scheme (DFS).} The DFS constructs the approximated process forward in time, starting from a trainable initial value $y$. 

\begin{enumerate}
    \item Initialize with $\YyZc_0 = y$.
    
    \item For $i=0,\ldots,N-1$, define recursively
    \eqlnostar{}{\YyZc_{i+1} = \YyZc_{i} - \feps\left(t_i,X_{t_i},\YyZc_{i},\Zc_i(X_{t_i}; \theta)\right) h + {\Zc_i(X_{t_i}; \theta)}\cdot \Delta W_{i}.}
    
    \item Train $(y, \theta)$ by minimizing the terminal loss
    $\Lc^{(F)}(y,\theta) := \Eb\big|\YyZc_N - \Phi(X_{t_N})\big|^2$,
    using stochastic optimization methods such as the Adam optimizer.

    \item  After training, the optimized value of $y$ is taken as the approximation of the initial value $Y_0$.
\end{enumerate}

\noindent
\textbf{Deep backward scheme (DBS).} The DBS starts from the terminal condition and propagates the approximated process backward along the time grid.

\begin{enumerate}

    \item Initialize with $\YcZc_N = \Phi(X_T).$
    
    \item For $i=N-1,\ldots,0$, define recursively
    \eqlnostar{eq-deep backward scheme}{\YcZc_{i} = \YcZc_{i+1} + \feps\left(t_{i+1},X_{t_{i+1}},\YcZc_{i+1},\Zc_i(X_{t_i}; \theta)\right) h - \Zc_i(X_{t_i}; \theta)\cdot \Delta W_{i}.}
    
   \item Train $\theta$ by minimizing the loss
    $\Lc^{(B)}(\theta) := \Eb\big|\YcZc_{0} - \Eb\big[\YcZc_{0}\big]\big|^2$,
    again using stochastic optimization methods such as Adam.

    \item After training, we use $\Eb\big[\YcZc_{0}\big]$ as the approximation of the initial value $Y_0$.

\end{enumerate}
In practice, the expectations in the loss functions are replaced by empirical averages over simulated mini-batches of trajectories of the forward process $X$ and Brownian increments. The network is then trained using these empirical losses. We note that the theoretical analysis is carried out for the population loss.

\section{Error Analysis}\label{sec:convergence analysis}
The analysis focuses on the DBS. We first derive the approximation error between the DBS solution $(\YcZc, \Zc)$ and the solution $(\Yeps,\Zeps)$ of the ReRBSDE \eqref{eq-ReRBSDE}, and then establish an upper bound for the DBS loss function $\Var(\YcZc_{0})$. The error analysis of the DFS can be obtained by adapting the arguments in \cite{han2020convergence} with suitable modifications. Hereafter, we assume $\eps<1$ and denote by $C$ a generic positive constant whose value may change from line to line but is independent of $h$ and $\eps$, and network parameters $\theta$. For brevity, in the proofs we omit the explicit time argument of $f$ whenever it is clear from the context, and use the subscript $i$ for $t_{i}$, $i=0,\ldots,N-1$. 

The main technical difficulty is that $\YcZc_i$, constructed in the DBS \eqref{eq-deep backward scheme}, is not $\Fc_{t_i}$-adapted; see also \citet[Remark 3.7]{gao2023convergence}. This prevents the direct use of the martingale orthogonality
\eqlnostar{}{\Eb\Big[\left(\YcZc_{i} - \Yeps_{t_i}\right)\int_{t_i}^{t_{i+1}}(\Zc_i(X_{t_i}) - \Zeps_s)\dd W_s\Big] = 0,}
which is a key ingredient in the classical numerical analysis of BSDE time discretization error, such as \citet[Theorem 5.3]{zhang2004numerical}. 

Our analysis addresses this issue by introducing the $\Fc_{t_i}$-adapted projection 
\eqlnostar{eq-Ycb definition}{\YcbZc_i := \Eb_i\left[\YcZc_{i}\right], i=0,\ldots,N-1, \quad \YcbZc_N = \YcZc_N = \Phi(X_{t_N}),}
which serves as an $\Fc_{t_i}$-measurable approximation of $Y_{t_i}$. We then quantify the projection residual $\YcZc_i - \YcbZc_i$ in terms of the training loss $\Var(\YcZc_0)$ and the time discretization error. This allows us to recover martingale orthogonality at the level of $\YcbZc_i$, while still controlling the original DBS variables $\YcZc$. 

Here we adopt the convention that $\Eb[\YcZc_{0}] = \Eb_0[\YcZc_{0}]$. 
Then, we can deduce, by the martingale representation theorem, that there exists a square-integrable and $\Fc$-adapted $\Zct$ such that
\eqlnostar{eq-martingale representation theorem of ycbzc org}{\YcbZc_{i+1} = \Eb_i\left[\YcbZc_{i+1}\right] + \int_{t_i}^{t_{i+1}}\Zct_s\dd W_s.}
Plug \eqref{eq-deep backward scheme} into the above equation, apply the law of iterated conditional expectations, and we have
\eqlnostar{eq-martingale representation theorem of ycbzc}{\YcbZc_{i+1} = \YcbZc_{i} - \Eb_i\left[\feps\left(t_{i+1},X_{t_{i+1}},\YcZc_{i+1},\Zc_i\right) h\right] + \int_{t_i}^{t_{i+1}}\Zct_s\dd W_s.}

\subsection{A Priori Error Estimate}
We begin with a lemma showing that the deviation of $\YcZc_{i}$ from its $\Lb^2(\Fc_{t_i})$-projection can be controlled by the variance of $\YcZc_0$, up to small remainder terms.
\begin{lemma}\label{lemma:estimate on L2 projection residual}
Assume \ref{hp-1}-\ref{hp-2}. Then, for any $\gamma_0>0$ and sufficiently small $h<\gamma_0$, the following estimates hold:
\eqlnostar{eq-estimate on L2 projection residual}{ \underset{0\leq i\leq N-1}{\max}\,\Eb\left|\YcZc_{i} - \YcbZc_{i}\right|^2 &\le \exp(C_1 T)\Var\left(\YcZc_{0}\right) + C\exp(C_1 T)\eps^{-2}h,\\
\label{eq-estimate on Zi Ztilde difference}
 \sum_{i=0}^{N-1}\Eb\int_{t_i}^{t_{i+1}}|\Zc_i - \Zct_s|^2\dd s &\le C_2\Var\left(\YcZc_{0}\right) + C\exp(C_1 T)\eps^{-2}h,}
where
\eqlnostar{eq-C1}{& C_1:=\frac{1}{\gamma_0}
+32\gamma_0\left(C_f^2+C_{\phi}^2\eps^{-2}\overline{\kappa}_{\infty}^2\right),\\
\label{eq-C2}
& C_2 := 1 + \exp(C_1 T)\Big(\frac{1}{\gamma_0}
+16\gamma_0\left(C_f^2+C_{\phi}^2\eps^{-2}\overline{\kappa}_{\infty}^2\right)\Big).}
\end{lemma}

\begin{proof}
For each $i=0,\ldots,N-1$, define $\Delta \YcbZc_{i} := \YcZc_i - \YcbZc_{i}$.
Subtracting \eqref{eq-martingale representation theorem of ycbzc} from \eqref{eq-deep backward scheme}, we obtain
\eqlnostar{eq-DBS Y l2 error intermediate 1}{\Delta \YcbZc_{i+1} = &\; \Delta \YcbZc_{i} - \feps\left(X_{i+1},\YcZc_{i+1},\Zc_i\right)h + \Zc_i\Delta W_i + \Eb_i\left[\feps\left(X_{i+1},\YcZc_{i+1},\Zc_i\right)h\right] - \int_{t_i}^{t_{i+1}}\Zct_s\dd W_s.}
Since $\Delta \YcbZc_{i+1}$ is orthogonal to $\int_{t_i}^{t_{i+1}}\Zc_i - \Zct_s\dd W_s$,  we square both sides and take expectation to obtain
\eqlnostar{}{\Eb\left|\Delta \YcbZc_{i+1}\right|^2 & + \Eb\int_{t_i}^{t_{i+1}}\left|\Zc_i - \Zct_s\right|^2\dd s \\
= & \Eb\left|\Delta \YcbZc_{i} - \feps\left(X_{i+1},\YcZc_{i+1},\Zc_i\right)h + \Eb_i\left[\feps\left(X_{i+1},\YcZc_{i+1},\Zc_i\right)h\right]\right|^2\\
\le &\; \Big(1+\frac{h}{\gamma_0}\Big)\Eb\left|\Delta \YcbZc_{i}\right|^2 + \left(1+\frac{\gamma_0}{h}\right)\Eb\left|\Eb_i\left[\feps\left(X_{i+1},\YcZc_{i+1},\Zc_i\right)h\right] - \feps\left(X_{i+1},\YcZc_{i+1},\Zc_i\right)h\right|^2 \\
\le &\; \Big(1+\frac{h}{\gamma_0}\Big)\Eb\left|\Delta \YcbZc_{i}\right|^2 + \left(1+\frac{\gamma_0}{h}\right)h^2\Eb\Big|\Eb_i\left[\feps\left(X_{i+1},\YcZc_{i+1},\Zc_i\right)\right] - \feps\left(X_{i},\Eb_i\left[\YcbZc_{i+1}\right],\Zc_i\right) \\
& + \feps\left(X_{i},\Eb_i\left[\YcbZc_{i+1}\right],\Zc_i\right) - \feps\left(X_{i+1},\YcZc_{i+1},\Zc_i\right)\Big|^2\\
\le &\; \Big(1+\frac{h}{\gamma_0}\Big)\Eb\left|\Delta \YcbZc_{i}\right|^2 + 2\gamma_0 h\Big(C h + 4C_f^2\Eb\left|\YcZc_{i+1} - \Eb_i\Big[\YcbZc_{i+1}\Big]\right|^2 \\
& + 4C_{\phi}^2\eps^{-2}\overline{\kappa}_{\infty}^2\Eb\left|\YcZc_{i+1} - \Eb_i\left[\YcbZc_{i+1}\right]\right|^2 + C\eps^{-2}h\Big) \\
\le &\; \Big(1+\frac{h}{\gamma_0}\Big)\Eb\left|\Delta \YcbZc_{i}\right|^2 + 8\gamma_0\left(C_f^2+C_{\phi}^2\eps^{-2}\overline{\kappa}_{\infty}^2\right)h\Big(C h + \Eb\left|\YcZc_{i+1} - \Eb_i\left[\YcbZc_{i+1}\right]\right|^2\Big) \\
\le &\; \Big(1+\frac{h}{\gamma_0}\Big)\Eb\left|\Delta \YcbZc_{i}\right|^2 + 8\gamma_0\left(C_f^2+C_{\phi}^2\eps^{-2}\overline{\kappa}_{\infty}^2\right)h\Big(C h + \Eb\Big|\YcZc_{i+1} - \YcbZc_{i+1} + \int_{t_i}^{t_{i+1}}\Zct_s\dd W_s\Big|^2\Big)\\
\label{eq-DBS Y l2 error intermediate 2}
\le &\; \Big(1+\frac{h}{\gamma_0}\Big)\Eb\left|\Delta \YcbZc_{i}\right|^2 + 16\gamma_0\left(C_f^2+C_{\phi}^2\eps^{-2}\overline{\kappa}_{\infty}^2\right)h\Big(C h + \Eb\int_{t_i}^{t_{i+1}}\left|\Zct_s\right|^2\dd s + \Eb\left|\Delta \YcbZc_{i+1}\right|^2\Big),
}
with $h < \gamma_0$, where the first inequality follows from Young's inequality, the third from the Lipschitz continuity of the generator $f$, the time regularity of $X$, the $1/2$-H\"older continuity of $f$ and $S$ in $t$, together with the regularity of $\phieps$ established in \autoref{lemma:regularity of phi^epsilon}, and the fifth from \eqref{eq-martingale representation theorem of ycbzc org}.

Now, rearranging terms yields
\eqlnostar{}{& \left(1-16\gamma_0\left(C_f^2+C_{\phi}^2\eps^{-2}\overline{\kappa}_{\infty}^2\right)h\right)\Eb\left|\Delta \YcbZc_{i+1}\right|^2 + \Eb\int_{t_i}^{t_{i+1}}\left|\Zc_i - \Zct_s\right|^2\dd s\label{eq-DBS Y l2 error intermediate 5}\\
\le &\; \Big(1+\frac{h}{\gamma_0}\Big)\Eb\left|\Delta \YcbZc_{i}\right|^2 + C\gamma_0\left(C_f^2+C_{\phi}^2\eps^{-2}\overline{\kappa}_{\infty}^2\right)h^2 + 16\gamma_0\left(C_f^2+C_{\phi}^2\eps^{-2}\overline{\kappa}_{\infty}^2\right)h\Eb\int_{t_i}^{t_{i+1}}\left|\Zct_s\right|^2\dd s.
}
Since $\left(1-16\gamma_0\left(C_f^2+C_{\phi}^2\eps^{-2}\overline{\kappa}_{\infty}^2\right)h\right)\left(1+32\gamma_0\left(C_f^2+C_{\phi}^2\eps^{-2}\overline{\kappa}_{\infty}^2\right)h\right)\ge1$ whenever \\
$32\gamma_0\left(C_f^2 +C_{\phi}^2\eps^{-2}\overline{\kappa}_{\infty}^2\right)h\le 1$, we obtain that, for small enough $h$,
\eqlnostar{}{\Eb\left|\Delta \YcbZc_{i+1}\right|^2 & + \left(1+32\gamma_0\left(C_f^2+C_{\phi}^2\eps^{-2}\overline{\kappa}_{\infty}^2\right)h\right)\Eb\int_{t_i}^{t_{i+1}}\left|\Zc_i - \Zct_s\right|^2\dd s\\
\le &\; \left(1+C_1 h\right)\Eb\left|\Delta \YcbZc_{i}\right|^2 + C\left(1+32\gamma_0\left(C_f^2+C_{\phi}^2\eps^{-2}\overline{\kappa}_{\infty}^2\right)h\right)\gamma_0\left(C_f^2+C_{\phi}^2\eps^{-2}\overline{\kappa}_{\infty}^2\right)h^2\\
\label{eq-DBS Y l2 error intermediate 3}
& + 16\left(1+32\gamma_0\left(C_f^2+C_{\phi}^2\eps^{-2}\overline{\kappa}_{\infty}^2\right)h\right)\gamma_0\left(C_f^2+C_{\phi}^2\eps^{-2}\overline{\kappa}_{\infty}^2\right)h\Eb\int_{t_i}^{t_{i+1}}\left|\Zct_s\right|^2\dd s,}
where
$C_1:=\frac{1}{\gamma_0}
+32\gamma_0\left(C_f^2+C_{\phi}^2\eps^{-2}\overline{\kappa}_{\infty}^2\right)$, as defined in \eqref{eq-C1}.

Then, by applying Gronwall's lemma (e.g., \citet[Proposition 3.2]{emmrich1999discrete}), we obtain that
\eqlnostar{}{
\underset{0\leq i\leq N-1}{\max}\,\Eb\left|\Delta \YcbZc_{i}\right|^2 \le &\; \exp(C_1 T)\Big(\Eb\left|\Delta \YcbZc_{0}\right|^2 + C\left(1+\eps^{-2}\right)h + C\left(1+\eps^{-2}\right)h\Eb\int_{0}^{T}\left|\Zct_s\right|^2\dd s\Big)\\
\label{eq-DBS Y l2 error intermediate 4}
\le &\; \exp(C_1 T)\Var\left(\YcZc_{0}\right) + C\exp(C_1 T)\eps^{-2}h.
}
Summing \eqref{eq-DBS Y l2 error intermediate 5} over $i=0,\ldots,N-1$ gives
\eqlnostar{}{& \left(1-16\gamma_0\left(C_f^2+C_{\phi}^2\eps^{-2}\overline{\kappa}_{\infty}^2\right)h\right) \sum_{i=0}^{N-1}\Eb\left|\Delta \YcbZc_{i+1}\right|^2 + \sum_{i=0}^{N-1}\Eb\int_{t_i}^{t_{i+1}}|\Zc_i - \Zct_s|^2\dd s\\
& \qquad \le \; \left(1+\frac{h}{\gamma_0}\right)\sum_{i=0}^{N-1}\Eb\left|\Delta \YcbZc_{i}\right|^2 + C\gamma_0\left(C_f^2+C_{\phi}^2\eps^{-2}\overline{\kappa}_{\infty}^2\right)h \\
&\qquad \quad + 16\gamma_0\left(C_f^2+C_{\phi}^2\eps^{-2}\overline{\kappa}_{\infty}^2\right)h\sum_{i=0}^{N-1}\Eb\int_{t_i}^{t_{i+1}}\left|\Zct_s\right|^2\dd s.
}
Therefore, by rearranging terms and applying \eqref{eq-DBS Y l2 error intermediate 4}, we obtain
\eqlnostar{}{
\sum_{i=0}^{N-1}\Eb\int_{t_i}^{t_{i+1}}|\Zc_i - \Zct_s|^2\dd s
\le & \; \Var\left( \YcZc_{0}\right) + \left(\frac{1}{\gamma_0}
+16\gamma_0\left(C_f^2+C_{\phi}^2\eps^{-2}\overline{\kappa}_{\infty}^2\right)\right)h\sum_{i=0}^{N-1}\Eb\left|\Delta \YcbZc_{i}\right|^2 + C\eps^{-2}h\\
\le & \; C_2\Var\left(\YcZc_{0}\right) + C\exp(C_1 T)\eps^{-2}h,
}
where
$C_2 := 1 + \exp(C_1 T)\left(\frac{1}{\gamma_0}
+16\gamma_0\left(C_f^2+C_{\phi}^2\eps^{-2}\overline{\kappa}_{\infty}^2\right)\right)$, as define in \eqref{eq-C2}.
\end{proof}

The following lemma will be used to control the error between $\YcbZc$ and $\Yeps$. 
\begin{lemma}\label{lemma:inequality for phi}
Assume that $h\overline{\kappa}_{\infty} C_{\phi}\eps^{-1}\leq 1,$ and \ref{hp-2} holds true, then $\forall t\in[0,T]$ and $\forall y\geq y'$, $-(y - y')\leq h\phieps(y - S_t)\kappa_t - h\phieps(y' - S_t)\kappa_t\leq 0,$ a.s.
\end{lemma}
\begin{proof}
Since $\phi^{\varepsilon}$ is non-increasing and its first order derivative is bounded by $C_{\phi}\eps^{-1}$, we have that $-C_{\phi}\eps^{-1}(y-y')\leq \phi^{\varepsilon}(y - S_t) - \phi^{\varepsilon}(y' - S_t)\leq 0$. Then we have the result since $\kappa_t\leq \overline{\kappa}_{\infty}, \forall t\in[0,T]$ a.s. by \ref{hp-2} and $h\overline{\kappa}_{\infty} C_{\phi}\eps^{-1}\leq 1$.
\end{proof}

\begin{theorem}\label{thm:deep backward scheme approx error}
Assume that $h\overline{\kappa}_{\infty} \frac{C_{\phi}}{\eps}\leq 1$, and \ref{hp-1}-\ref{hp-2} hold true. For some $\gamma_1 > h$ such that $1-40\gamma_1 C_f^2 > 0$, we have
\eqlnostar{}{\underset{0\leq i\leq N-1}{\max}\,\Eb\left|\YcbZc_{i} - \Yeps_{t_i}\right|^2 \le &\; C\exp(C_1 T + C_3 T)\eps^{-4}h\\
& + 20\gamma_1\exp(C_3 T)\left(\left(C_f^2+C_{\phi}^2\eps^{-2}\overline{\kappa}_{\infty}^2\right)\exp(C_1 T) + 2C_f^2 C_2\right)\Var\left(\YcZc_0\right),\\
\Eb\int_{0}^{T}\left|\Zct_t - \Zeps_t\right|^2\dd t
\le & \; C\exp(C_1 T + C_3 T)\eps^{-4}h + C_4\Var\left(\YcZc_{0}\right).} 
Here $C_1$ and $C_2$ are defined in \eqref{eq-C1} and \eqref{eq-C2}, respectively, and
\eqlnostar{eq-C3}{C_3 := &\; \frac{1}{\gamma_1} + 10C_f^2\gamma_1,\\
C_4 := &\; \left(1-40\gamma_1 C_f^2\right)^{-1}\Big(20\exp(C_3 T)\gamma_1\left(\left(C_f^2+C_{\phi}^2\eps^{-2}\overline{\kappa}_{\infty}^2\right)\exp(C_1 T) + C_f^2 C_2\right) \\
\label{eq-C4}
& + 20\left(C_f^2+C_{\phi}^2\eps^{-2}\overline{\kappa}_{\infty}^2\right)\gamma_1 \exp(C_1 T) + 40\gamma_1 C_f^2C_2\Big).}
\end{theorem}

\begin{proof}
For each $i=0,\ldots,N-1$, define 
$\Delta \Ycb_{i} := \YcbZc_i - \Yeps_{i}$, and $\Delta \Zct_t := \Zct_{t} - \Zeps_t$. Taking the difference between $\YcbZc$ and $\Yeps$, we have
\eqlnostar{}{\Delta \Ycb_{i} = \Delta \Ycb_{i+1} + \int_{t_i}^{t_{i+1}}\Eb_i\left[\feps\left(X_{i+1},\YcZc_{i+1},\Zc_i\right)\right] - \feps\left(X_s,\Yeps_s,\Zeps_s\right)\dd s - \int_{t_i}^{t_{i+1}}\Delta \Zct_s\dd W_s.}
Since $\Delta \Ycb_i$ is $\Fc_{t_i}$-measurable, we rearrange the terms, square both sides, and take expectations to obtain
\eqlnostar{}{\Eb\left|\Delta \Ycb_{i}\right|^2 + \Eb\int_{t_i}^{t_{i+1}}|\Delta \Zct_s|^2\dd s = \Eb\Big|\Delta \Ycb_{i+1} + \int_{t_i}^{t_{i+1}}\Eb_i\left[\feps\left(X_{i+1},\YcZc_{i+1},\Zc_i\right)\right] - \feps\left(X_s,\Yeps_s,\Zeps_s\right)\dd s\Big|^2.}
Then, it can be derived that
\eqlnostar{}{\Eb\left|\Delta \Ycb_{i}\right|^2 & + \Eb\int_{t_i}^{t_{i+1}}|\Delta \Zct_s|^2\dd s \le \Eb\Big(\left|\Delta \Ycb_{i+1} + \phieps\left(\YcbZc_{i+1} - S_{i+1}\right)\kappa_{i+1}h - \phieps(\Yeps_{i+1} - S_{i+1})\kappa_{i+1}h\right| \\
& + \Big|- \phieps\left(\YcbZc_{i+1} - S_{i+1}\right)\kappa_{i+1}h + \phieps(\Yeps_{i+1} - S_{i+1})\kappa_{i+1}h \\
& + \int_{t_i}^{t_{i+1}}\Eb_i\left[\feps\left(X_{i+1},\YcZc_{i+1},\Zc_i\right)\right] - \feps\left(X_s,\Yeps_s,\Zeps_s\right)\dd s\Big|\Big)^2\\
\le & \; \left(1+\frac{h}{\gamma_1}\right)\Eb\left|\Delta \Ycb_{i+1} + \phieps\left(\YcbZc_{i+1} - S_{i+1}\right)\kappa_{i+1}h - \phieps(\Yeps_{i+1} - S_{i+1})\kappa_{i+1}h\right|^2 \\
& + \left(1+\frac{\gamma_1}{h}\right)\Eb\Big|- \phieps\left(\YcbZc_{i+1} - S_{i+1}\right)\kappa_{i+1}h
+ \phieps(\Yeps_{i+1} - S_{i+1})\kappa_{i+1}h \\
& + \int_{t_i}^{t_{i+1}}\Eb_i\left[\feps\left(X_{i+1},\YcZc_{i+1},\Zc_i\right)\right] - \feps\left(X_s,\Yeps_s,\Zeps_s\right)\dd s\Big|^2\\
\le & \; \left(1+\frac{h}{\gamma_1}\right)\Eb\left|\Delta \Ycb_{i+1}\right|^2 + \left(1+\frac{\gamma_1}{h}\right)\Eb\Big|- \phieps\left(\YcbZc_{i+1} - S_{i+1}\right)\kappa_{i+1}h
+ \phieps(\Yeps_{i+1} - S_{i+1})\kappa_{i+1}h \\
& + \int_{t_i}^{t_{i+1}}\Eb_i\left[\feps\left(X_{i+1},\YcZc_{i+1},\Zc_i\right)\right] - \feps\left(X_s,\Yeps_s,\Zeps_s\right)\dd s\Big|^2\\
\le & \; \left(1+\frac{h}{\gamma_1}\right)\Eb\left|\Delta \Ycb_{i+1}\right|^2 + \left(1+\frac{\gamma_1}{h}\right)\Eb\Big|- \phieps\left(\YcbZc_{i+1} - S_{i+1}\right)\kappa_{i+1}h
+ \phieps(\Yeps_{i+1} - S_{i+1})\kappa_{i+1}h \\
& + \feps\left(X_{i},\YcbZc_{i},\Zc_i\right)h - \feps\left(X_{i},\YcbZc_{i},\Zc_i\right)h + f\left(X_{i+1},\YcbZc_{i+1},\Zc_i\right)h - f\left(X_{i+1},\YcbZc_{i+1},\Zc_i\right)h\\
& + \Eb_i\left[\feps\left(X_{i+1},\YcbZc_{i+1},\Zc_i\right)h - \feps\left(X_{i+1},\YcbZc_{i+1},\Zc_i\right)h\right] + f\left(X_{i+1},\Yeps_{i+1},\Zc_i\right)h \\
& - f\left(X_{i+1},\Yeps_{i+1},\Zc_i\right)h + \int_{t_i}^{t_{i+1}}\Eb_i\left[\feps\left(X_{i+1},\YcZc_{i+1},\Zc_i\right)\right] - \feps\left(X_s,\Yeps_s,\Zeps_s\right)\dd s\Big|^2\\
\le & \; \left(1+\frac{h}{\gamma_1}\right)\Eb\left|\Delta \Ycb_{i+1}\right|^2 + 5\left(1+\frac{\gamma_1}{h}\right)\Eb\Big[\left|\feps\left(X_{i},\YcbZc_{i},\Zc_i\right) - \feps\left(X_{i+1},\YcbZc_{i+1},\Zc_i\right)\right|^2 h^2 \\
& + \Eb_i\left|\feps\left(X_{i+1},\YcbZc_{i+1},\Zc_i\right) - \feps\left(X_{i},\YcbZc_{i},\Zc_i\right)\right|^2 h^2 \\
& + \left|f\left(X_{i+1},\YcbZc_{i+1},\Zc_i\right) - f\left(X_{i+1},\Yeps_{i+1},\Zc_i\right)\right|^2 h^2\\
& + \Eb_i\left|\feps\left(X_{i+1},\YcZc_{i+1},\Zc_i\right) - \feps\left(X_{i+1},\YcbZc_{i+1},\Zc_i\right)\right|^2 h^2 \\
& + \left|\int_{t_i}^{t_{i+1}}\left(\feps\left(X_{i+1},\Yeps_{i+1},\Zc_i\right) - \feps\left(X_s,\Yeps_s,\Zeps_s\right)\right)\dd s\right|^2\Big]\\
\le & \; \left(1+C_3 h\right)\Eb\left|\Delta \Ycb_{i+1}\right|^2 + 10\frac{\gamma_1}{ h}\Big(\underbrace{\Eb\left|\feps\left(X_{i},\YcbZc_{i},\Zc_i\right) - \feps\left(X_{i+1},\YcbZc_{i+1},\Zc_i\right)\right|^2}_{{I}_i} h^2 \\
& + \underbrace{\Eb\left|\feps\left(X_{i+1},\YcbZc_{i+1},\Zc_i\right) - \feps\left(X_{i},\YcbZc_{i},\Zc_i\right)\right|^2}_{{II}_i} h^2 \\
& + \underbrace{\Eb\left|\feps\left(X_{i+1},\YcZc_{i+1},\Zc_i\right) - \feps\left(X_{i+1},\YcbZc_{i+1},\Zc_i\right)\right|^2}_{{III}_i} h^2\\
\label{eq-DBS approx error intermediate 1}
& + \underbrace{\Eb\left|\int_{t_i}^{t_{i+1}}\left(\feps\left(X_{i+1},\Yeps_{i+1},\Zc_i\right) - \feps\left(X_s,\Yeps_s,\Zeps_s\right)\right)\dd s\right|^2}_{{IV}_i}\Big),}
where
$C_3 := \frac{1}{\gamma_1} + 10C_f^2\gamma_1$ as defined in \eqref{eq-C3}, with $h < \gamma_1$. The second inequality follows from Young's inequality, the third from \autoref{lemma:inequality for phi}, the fifth from a rearrangement of terms, and the sixth from Jensen's inequality together with the Lipschitz continuity of the generator $f$.

We next show that the terms ${I}_i, {II}_i, {III}_i, {IV}_i$ in \eqref{eq-DBS approx error intermediate 1} satisfy:
\eqlnostar{}{& {I}_i\le C\eps^{-2}h + C\eps^{-2}\Eb\int_{t_i}^{t_{i+1}}\left|\Zct_s\right|^2\dd s,\\
& {II}_i\le C\eps^{-2}h + C\eps^{-2}\Eb\int_{t_i}^{t_{i+1}}\left|\Zct_s\right|^2\dd s,\\
& {III}_i\le 2\left(C_f^2+C_{\phi}^2\eps^{-2}\overline{\kappa}_{\infty}^2\right)\Eb\left|\YcZc_{i+1} - \YcbZc_{i+1}\right|^2, \\
& {IV}_i\le 4C_f^2 h \Eb\int_{t_i}^{t_{i+1}}\Big(\left|\Zc_i - \Zct_s\right|^2 + \left|\Zct_s - \Zeps_s\right|^2\Big)\dd s + C\left(h^3 + \eps^{-2}h^{3}\right).
}
To derive the bounds for ${I}_i$ and ${II}_i$, note that 
\eqlnostar{}{& \Eb\left|\feps\left(X_{i+1},\YcbZc_{i+1},\Zc_i\right) - \feps\left(X_{i},\YcbZc_{i},\Zc_i\right)\right|^2 \\
\le &\; 2\Eb\Big(\left|f\left(X_{i+1},\YcbZc_{i+1},\Zc_i\right) - f\left(X_{i},\YcbZc_{i},\Zc_i\right)\right|^2 + \left|\phieps\left(\YcbZc_{i+1}-S_{i+1}\right)\kappa_{i+1} - \phieps\left(\YcbZc_{i}-S_{i}\right)\kappa_i\right|^2\Big) \\
\le &\; \Eb\Big(2\left|f\left(X_{i+1},\YcbZc_{i+1},\Zc_i\right) - f\left(X_{i},\YcbZc_{i},\Zc_i\right)\right|^2 + 4\left|\phieps\left(\YcbZc_{i+1}-S_{i+1}\right) - \phieps\left(\YcbZc_{i}-S_{i}\right)\right|^2|\kappa_{i+1}|^2\Big) \\
& + 4\Eb\Big(\left|\phieps\left(\YcbZc_{i}-S_{i}\right)\right|^2|\kappa_{i+1} - \kappa_{i}|^2\Big) \\
\le &\, \Eb\Big(4C_f^2\left|X_{i+1} - X_{i}\right|^2 + 4C_f^2\left|\YcbZc_{i+1} - \YcbZc_{i}\right|^2 + 8C_{\phi}^2\eps^{-2}\overline{\kappa}_{\infty}^2\left|\YcbZc_{i+1} - \YcbZc_{i}\right|^2 + 8C_{\phi}^2\eps^{-2}\overline{\kappa}_{\infty}^2\left|S_{i+1} - S_{i}\right|^2\Big)\\
& + 4C_f^2\Eb\Big(|X_{i+1} - X_{i}|^2 + |S_{i+1} - S_{i}|^2 + |V_{i+1} - V_{i}|^2\Big) + 4\Eb|U_{i+1} - U_{i}|^2 \\
\label{eq-DBS approx error intermediate 7}
\le &\, C(1+\eps^{-2})h + \left(4C_f^2 + 8C_{\phi}^2\eps^{-2}\overline{\kappa}_{\infty}^2\right)\Eb\left|\YcbZc_{i+1} - \YcbZc_{i}\right|^2\\
\le &\, C(1+\eps^{-2})h + \left(4C_f^2 + 8C_{\phi}^2\eps^{-2}\overline{\kappa}_{\infty}^2\right)\Eb\Big(\left|\feps\left(X_{t_{i+1}},\YcZc_{i+1},\Zc_i\right) h\right|^2 + \int_{t_i}^{t_{i+1}}\left|\Zct_s\right|^2\Big)\dd s\\
\le &\, C\eps^{-2}h + \left(4C_f^2 + 8C_{\phi}^2\eps^{-2}\overline{\kappa}_{\infty}^2\right)\Eb\int_{t_i}^{t_{i+1}}\left|\Zct_s\right|^2\dd s,}
where the third inequality follows from the Lipschitz continuity of the generator $f$, together with the regularity of $\phieps$ established in \autoref{lemma:regularity of phi^epsilon}, the fourth from the regularity of $X$, $S$, $U$ and $V$, and the fifth from \eqref{eq-martingale representation theorem of ycbzc}. The estimate for ${III}_i$ can be derived in a similar manner. As for ${IV}_i$, we have
\eqlnostar{}{& \Eb\Big|\int_{t_i}^{t_{i+1}}\Big(\feps\left(X_{i+1},\Yeps_{i+1},\Zc_i\right) - \feps\left(X_s,\Yeps_s,\Zeps_s\right)\Big)\dd s\Big|^2\\
\le &\; 2\Eb\Big|\int_{t_i}^{t_{i+1}}\left(\feps\left(X_{i+1},\Yeps_{i+1},\Zc_i\right) - \feps\left(X_{i+1},\Yeps_{i+1},\Zeps_s\right)\right)\dd s\Big|^2\\
& + 2\Eb\Big|\int_{t_i}^{t_{i+1}}\left(\feps\left(X_{i+1},\Yeps_{i+1},\Zeps_s\right) - \feps\left(X_s,\Yeps_s,\Zeps_s\right)\right)\dd s\Big|^2\\
\le &\; 4C_f^2 h\Eb\int_{t_i}^{t_{i+1}}\left|\Zc_i - \Zct_s\right|^2 + \left|\Zct_s - \Zeps_s\right|^2\dd s + 4h\Eb\int_{t_i}^{t_{i+1}}\Big(\Big|\left(f\left(X_{i+1},\Yeps_{i+1},\Zeps_s\right) - f\left(X_s,\Yeps_s,\Zeps_s\right)\right)\Big|^2\\
& + \left|\phieps\left(\Yeps_{i+1}-S_{i+1}\right)\kappa_{i+1} - \phieps\left(\Yeps_{s}-S_{s}\right)\kappa_s\right|^2\Big)\dd s \\
\le &\, 4C_f^2 h\Eb\int_{t_i}^{t_{i+1}}\left|\Zc_i - \Zct_s\right|^2 + \left|\Zct_s - \Zeps_s\right|^2\dd s + 8C_f^2 h\Eb\int_{t_i}^{t_{i+1}}\Big(\left|X_{i+1} - X_{s}\right|^2 + \left|\Yeps_{i+1} - \Yeps_{s}\right|^2\Big)\dd s\\
& + 16C_{\phi}^2\eps^{-2}\overline{\kappa}_{\infty}^2h\Eb\int_{t_i}^{t_{i+1}}\Big(\left|\Yeps_{i+1} - \Yeps_{s}\right|^2 + \left|S_{i+1} - S_{s}\right|^2\Big)\dd s \\
& + 8C_f^2h\Eb\int_{t_i}^{t_{i+1}}\Big(|X_{i+1} - X_{s}|^2 + |S_{i+1} - S_{s}|^2 + |V_{i+1} - V_{s}|^2\Big)\dd s + 8h\Eb\int_{t_i}^{t_{i+1}}|U_{i+1} - U_{s}|^2\dd s \\
\le &\; 4C_f^2 h\Eb\int_{t_i}^{t_{i+1}}\left|\Zc_i - \Zct_s\right|^2 + \left|\Zct_s - \Zeps_s\right|^2\dd s + C(h^3 + \eps^{-2}h^3),
}
where the second inequality follows from the Lipschitz continuity of $f$, and the third and fourth inequalities are obtained by arguments analogous to those used for \eqref{eq-DBS approx error intermediate 7}.

Substituting the estimates for ${I}_i, {II}_i, {III}_i, {IV}_i$ into \eqref{eq-DBS approx error intermediate 1}, we obtain
\eqlnostar{}{& \Eb\left|\Delta \Ycb_{i}\right|^2 + \left(1-40\gamma_1 C_f^2\right)\Eb\int_{t_i}^{t_{i+1}}|\Delta \Zct_s|^2\dd s \\
\le & \; \left(1+C_3 h\right)\Eb\left|\Delta \Ycb_{i+1}\right|^2 + C\eps^{-2}h\Eb\int_{t_i}^{t_{i+1}}\left|\Zct_s\right|^2\dd s + C(h+\gamma_1)\left(\eps^{-2}h^2 + h^2\right) \\
\label{eq-DBS approx error intermediate 2}
& + 20\left(C_f^2+C_{\phi}^2\eps^{-2}\overline{\kappa}_{\infty}^2\right)\gamma_1 h \Eb\left|\YcZc_{i+1} - \YcbZc_{i+1}\right|^2 + 40\gamma_1 C_f^2\Eb\int_{t_i}^{t_{i+1}}\left|\Zc_i - \Zct_s\right|^2\dd s.}
Then choose $\gamma_1>0$ such that $1-40\gamma_1 C_f^2 >0$. Applying (backward) Gronwall's lemma (cf. \citet[Theorem 5.3.3]{zhang2017backward}), we obtain
\eqlnostar{}{
& \underset{0\leq i\leq N-1}{\max}\,\Eb\left|\Delta \Ycb_{i}\right|^2\\
\le &\; \exp(C_3 T)\Eb\left|\Delta \Ycb_{N}\right|^2 + C\exp(C_3 T)\left(\eps^{-2}h + h\right) + 40\exp(C_3 T)C_f^2\gamma_1 \sum_{i=0}^{N-1} \Eb\int_{t_i}^{t_{i+1}}\left|\Zc_i - \Zct_s\right|^2\dd s\\
& + 20\exp(C_3 T)\left(C_f^2+C_{\phi}^2\eps^{-2}\overline{\kappa}_{\infty}^2\right)\gamma_1 h\sum_{i=0}^{N-1} \Eb\left|\YcZc_{i+1} - \YcbZc_{i+1}\right|^2\\
\le &\; C\exp(C_3 T)\left(\eps^{-2}h + h\right) + 20\exp(C_3 T)\left(C_f^2+C_{\phi}^2\eps^{-2}\overline{\kappa}_{\infty}^2\right)\gamma_1 h\sum_{i=0}^{N-1} \Eb\left|\YcZc_{i+1} - \YcbZc_{i+1}\right|^2\\
\label{eq-DBS approx error intermediate 3}
& + 40\exp(C_3 T)C_f^2\gamma_1 \sum_{i=0}^{N-1} \Eb\int_{t_i}^{t_{i+1}}\left|\Zc_i - \Zct_s\right|^2\dd s.
}
Now combining results from \autoref{lemma:estimate on L2 projection residual}, we prove that
\eqlnostar{}{\underset{0\leq i\leq N-1}{\max}\,\Eb\left|\Delta \Ycb_{i}\right|^2 \le &\; C\exp(C_1 T + C_3 T)\left(\eps^{-2}h + h + \eps^{-4}h\right)\\
& + 20\gamma_1\exp(C_3 T)\left(\left(C_f^2+C_{\phi}^2\eps^{-2}\overline{\kappa}_{\infty}^2\right)\exp(C_1 T) + 2C_f^2 C_2\right)\Var\left(\YcZc_0\right).}

Summing \eqref{eq-DBS approx error intermediate 2} over $i=0,\ldots,N-1$ gives
\eqlnostar{}{\sum_{i=0}^{N-1}\Eb\left|\Delta \Ycb_{i}\right|^2 & + \left(1-40\gamma_1 C_f^2\right)\sum_{i=0}^{N-1}\Eb\int_{t_i}^{t_{i+1}}|\Delta \Zct_s|^2\dd s
\le \left(1+C_3 h\right)\sum_{i=0}^{N-1}\Eb\left|\Delta \Ycb_{i+1}\right|^2  + C\eps^{-2}h\Eb\int_{0}^{T}\left|\Zct_s\right|^2\dd s\\
& + C\sum_{i=0}^{N-1}\left(\eps^{-2}h^2 + h^2\right) + 20\left(C_f^2+C_{\phi}^2\eps^{-2}\overline{\kappa}_{\infty}^2\right)\gamma_1 h\sum_{i=0}^{N-1}\Eb\left|\YcZc_{i+1} - \YcbZc_{i+1}\right|^2\\
\label{eq-DBS approx error intermediate 4}
& + 40\gamma_1 C_f^2\sum_{i=0}^{N-1}\Eb\int_{t_i}^{t_{i+1}}\left|\Zc_i - \Zct_s\right|^2\dd s.}
Rearranging the terms, we get
\eqlnostar{}{& \left(1-40\gamma_1 C_f^2\right)\sum_{i=0}^{N-1}\Eb\int_{t_i}^{t_{i+1}}|\Delta \Zct_s|^2\dd s + \Eb\left|\Delta \Ycb_{0}\right|^2\\
\le &\; \left(1+C_3 h\right)\Eb\left|\Delta \Ycb_{N}\right|^2 + C_3 h\sum_{i=1}^{N-1}\Eb\left|\Delta \Ycb_{i}\right|^2 + C\left(\eps^{-2}h + h\right) \\
& + 20\left(C_f^2+C_{\phi}^2\eps^{-2}\overline{\kappa}_{\infty}^2\right)\gamma_1 h\sum_{i=0}^{N-1}\Eb\left|\YcZc_{i+1} - \YcbZc_{i+1}\right|^2 + 40\gamma_1 C_f^2\sum_{i=0}^{N-1}\Eb\int_{t_i}^{t_{i+1}}\left|\Zc_i - \Zct_s\right|^2\dd s\\
\le & \; C_3 h\sum_{i=1}^{N-1}\Eb\left|\Delta \Ycb_{i}\right|^2 + C\left(\eps^{-2}h + h\right) \\
& + 20\left(C_f^2+C_{\phi}^2\eps^{-2}\overline{\kappa}_{\infty}^2\right)\gamma_1 h\sum_{i=0}^{N-1}\Eb\left|\YcZc_{i+1} - \YcbZc_{i+1}\right|^2 + 40\gamma_1 C_f^2\sum_{i=0}^{N-1}\Eb\int_{t_i}^{t_{i+1}}\left|\Zc_i - \Zct_s\right|^2\dd s\\
\le & \; C\exp(C_1 T + C_3 T)\left(\eps^{-2}h + h + \eps^{-4}h\right) \\
& + 20\exp(C_3 T)\gamma_1\left(\left(C_f^2+C_{\phi}^2\eps^{-2}\overline{\kappa}_{\infty}^2\right)\exp(C_1 T)+2C_f^2 C_2\right)\Var\left(\YcZc_0\right)\\
\label{eq-DBS approx error intermediate 5}
& + 20\left(C_f^2+C_{\phi}^2\eps^{-2}\overline{\kappa}_{\infty}^2\right)\gamma_1 h\sum_{i=0}^{N-1}\Eb\left|\YcZc_{i+1} - \YcbZc_{i+1}\right|^2 + 40\gamma_1 C_f^2\sum_{i=0}^{N-1}\Eb\int_{t_i}^{t_{i+1}}\left|\Zc_i - \Zct_s\right|^2\dd s.}
Combining this estimate once more with the bounds from \autoref{lemma:estimate on L2 projection residual}, we obtain the second statement of our main result:
\eqlnostar{eq-DBS approx error intermediate 6}{\Eb\int_{0}^{T}|\Delta \Zct_s|^2\dd s
\le & \; C\exp(C_1 T + C_3 T)\left(\eps^{-2}h + h + \eps^{-4}h\right) + C_4\Var\left(\YcZc_{0}\right),}
with $C_4$ defined in \eqref{eq-C4}. Since $\eps<1$, the proof is complete.
\end{proof}

\begin{corollary}\label{corollary:DBS approximation error}
Assume the same assumptions as in \autoref{thm:deep backward scheme approx error}, and let $\YcbZc_i, i=0,\ldots,N$ be defined as in \eqref{eq-Ycb definition}. Then, an immediate result by orthogonality of the $\Lb^2(\Fc_{t_i})$ projection is that:
\eqlnostar{eq-MSE decomposition}{
\Eb\left|\YcZc_i - \Yeps_{t_i}\right|^2 = \Eb\left|\YcZc_i - \YcbZc_i\right|^2 + \Eb\left|\YcbZc_i - \Yeps_{t_i}\right|^2.
} 
Therefore, we have
\eqlnostar{}{\underset{0\leq i\leq N-1}{\max}\,\Eb\Big|\YcZc_i & - \Yeps_{t_i}\Big|^2 \le 
C\exp(C_1 T + C_3 T)\eps^{-4}h\\
& + \left(20\gamma_1\exp(C_3 T)\left(\left(C_f^2+C_{\phi}^2\eps^{-2}\overline{\kappa}_{\infty}^2\right)\exp(C_1 T) + 2C_f^2 C_2\right) + \exp(C_1 T)\right)\Var\left(\YcZc_0\right).}
Furthermore, \autoref{lemma:estimate on L2 projection residual} and \autoref{thm:deep backward scheme approx error} yield that
\eqlnostar{}{\sum_{i=0}^{N-1}\Eb\int_{t_i}^{t_{i+1}}|\Zc_i - \Zeps_t|^2\dd t \le & \; 2(C_2 + C_4)\Var\left(\YcZc_{0}\right) + C\exp(C_1 T + C_3 T)\eps^{-4}h.}
\end{corollary}

Combining \autoref{thm:a priori}, \autoref{lemma:estimate on L2 projection residual}, \autoref{thm:deep backward scheme approx error}, and \autoref{corollary:DBS approximation error}, we have the following upper bound for the error between the approximation $(\YcZc,\Zc)$ and the solution $(Y,Z)$ of RBSDE \eqref{eq-RBSDE}.

\begin{theorem}[Approximation error of DBS]\label{thm:true DBS approx error}
Assume the same assumptions as in \autoref{thm:deep backward scheme approx error}, and let $(Y,Z)\in\Sb^2\times\Hb^2$ be the solution of RBSDE \eqref{eq-RBSDE}. Then, we have
\eqlnostar{}{ \underset{0\leq i\leq N-1}{\max}\,\Eb\Big|\YcZc_{i} - Y_{t_i}\Big|^2 &\le C\eps + C\exp(C_1 T + C_3 T)\eps^{-4}h\\
 &\hspace{-60pt}+ \left(40\gamma_1\exp(C_3 T)\left(\left(C_f^2+C_{\phi}^2\eps^{-2}\overline{\kappa}_{\infty}^2\right)\exp(C_1 T) + 2C_f^2 C_2\right)+\exp(C_1T)\right)\Var\left(\YcZc_0\right),\\
 \sum_{i=0}^{N-1}\Eb\int_{t_i}^{t_{i+1}}|\Zc_i - Z_t|^2\dd t &\le C\eps + 3(C_2+C_4)\Var\left(\YcZc_{0}\right) + C\exp(C_1 T + C_3 T)\eps^{-4}h,} 
where $C_1, C_2, C_3, C_4$ are defined in \eqref{eq-C1}, \eqref{eq-C2}, \eqref{eq-C3}, and \eqref{eq-C4}, respectively. Additionally, if $f$ is non-increasing in $y$, then
\eqlnostar{}{& \underset{0\leq i\leq N-1}{\max}\,\Eb\Big|\YcZc_{i} - Y_{t_i}\Big|^2 \le C\eps^2 + C\exp(C_1 T + C_3 T)\eps^{-4}h\\
& + \left(40\gamma_1\exp(C_3 T)\left(\left(C_f^2+C_{\phi}^2\eps^{-2}\overline{\kappa}_{\infty}^2\right)\exp(C_1 T) + 2C_f^2 C_2\right)+\exp(C_1T)\right)\Var\left(\YcZc_0\right).}
\end{theorem} 

\begin{remark}
\autoref{corollary:DBS approximation error} shows that the distance between the regularized solution $(\Yeps,\Zeps)$ and the DBS approximation $(\YcZc,\Zc)$ can be controlled by the variance of the initial approximation $\YcZc_0$. A notable feature of our analysis is that, unlike \cite{gao2023convergence}, the error estimates does not impose an a priori smallness assumption on the Lipschitz constants of the driver with respect to $y$ and $z$. Thus, the argument applies to a broader class of drivers.

\autoref{thm:true DBS approx error} combines this estimate with the regularization error bound established in \autoref{thm:a priori} to obtain an error estimate between the DBS approximation $(\YcZc,\Zc)$ and the solution $(Y,Z)$ of the original RBSDE. The result also reveals a trade-off in the choice of the regularization parameter $\eps$. On the one hand, $\eps$ should be small so that the regularized solution $(\Yeps,\Zeps)$ is close to the original RBSDE solution $(Y,Z)$. On the other hand, $\eps$ cannot be chosen too small relative to the time step $h$, since the DBS discretization error contains the term $\exp(C_1T+C_3T)\eps^{-4}h$. By the definition of $C_1$ in \eqref{eq-C1}, this term behaves like, up to constants,
$\exp(\overline C\eps^{-2})\eps^{-4}h$. Therefore, although decreasing $\eps$ reduces the regularization error, it also amplifies the approximation error exponentially. This suggests, at the level of the displayed upper bound, that the choice of $\eps$ should balance the regularization error and the DBS approximation error.

The exponential dependence on $\eps^{-2}$ leads formally to a logarithmic choice of the regularization parameter. Indeed, for any $\alpha>\overline C$, taking $\eps=\sqrt{\alpha|\log h|^{-1}}$ gives $\exp(\overline C\eps^{-2})\eps^{-4}h
=
\alpha^{-2}h^{1-\overline C/\alpha}|\log h|^2
\to 0$.
This calculation should be interpreted only as guidance from the displayed upper bound, since the remaining constants and stability conditions in the error estimate are not uniform in $\eps$. In the numerical experiments, we instead follow the empirical choice $\eps=h$, which has been observed to perform well in \cite{agarwal2026numerical}. We adopt this choice throughout Section~\ref{sec:numerical results}.
\end{remark}

\subsection{An Upper Bound for the Training Objective}

To complement the error bounds established above, we now relate the training objective $\Var(\YcZc_0)$ to the universality of neural networks. The representation of $\Zeps$ used for this purpose is stated in the following proposition. It is a direct byproduct of the proof of \cite[Theorem 4]{bally2002reflected}, applied to the semilinear PDE associated with the regularized RBSDE. In particular, the result shows that $\Zeps_t=\sigma(t,X_t)^\top\partial_x v^\eps(t,X_t)$, where $v^\eps$ is the weak solution of the associated semilinear PDE. Hence, $\Zeps$ can be viewed as a deterministic function of time and the forward state process, which allows us to connect the DBS training objective to the neural-network approximation error for the function representing $\Zeps$.

\begin{proposition}[Markovian representation of the regularized RBSDE]
\label{propositon:nonlinear Feynman-Kac formula}
Let $\rho:\Rb^d\to\Rb_+$ be a continuous, strictly positive weight function
such that $\int_{\Rb^d}\rho(x)\,\dd x<\infty$, and set $L^2_\rho:=L^2(\Rb^d,\rho(x)\dd x).$ Define
\eqlnostar{}{H^1_\rho := \Bigl\{u:[0,T]\times\Rb^d\to\Rb: u_t\in L^2_\rho,\ 0\le t\le T,\ \sigma^\top \partial_x u \in L^2([0,T]\times\Rb^d,\dd t\times\rho(x)\dd x)\Bigr\}.
}
Let the obstacle be given by $\Phi(t,x):=h(t,x)+q(t,x)$, where $h$ is sufficiently smooth, $q$ is measurable, and both $h$ and $q$ have polynomial growth. Assume further that, for every $(t,x)$, the process $s\mapsto q(s,X_s)$ is a square-integrable c\`adl\`ag submartingale. Define the Markovian nonlinearity 
\eqlnostar{}{
F^\eps(t,x,y,z) := f(t,x,y,z) + \phieps\bigl(y-\Phi(t,x)\bigr)\Bigl(f\bigl(t,x,h(t,x),\sigma(t,x)^\top\partial_x h(t,x)\bigr) + (\partial_t+\mathcal L)h(t,x)\Bigr)^-.}
Here $\mathcal L\psi(t,x)
:=
\sum_{i,j=1}^d \frac12(\sigma\sigma^\top)_{ij}(t,x)\partial_{x_i x_j}^2\psi(t,x)
+\sum_{i=1}^d b_i(t,x)\partial_{x_i}\psi(t,x)$. Furthermore, assume that $v^\eps\in H^1_\rho$ satisfies $v^\eps(T,x)=\Phi(T,x),\; v^\eps(t,x)\ge \Phi(t,x), (t,x)\in[0,T]\times\Rb^d$, and, for every smooth test function $\varphi$ with compact support in $[0,T]\times\Rb^d$,
\eqlnostar{}{
&\int_t^T\int_{\Rb^d} v^\eps(s,x)\partial_s\varphi(s,x)\,\dd x\dd s -\int_{\Rb^d} \Phi(T,x)\varphi(T,x)\,\dd x +\int_{\Rb^d} v^\eps(t,x)\varphi(t,x)\,\dd x
\\
&\quad +\int_t^T\int_{\Rb^d} \Bigl( (\sigma^\top\partial_x v^\eps)(s,x) (\sigma^\top\partial_x\varphi)(s,x) + v^\eps(s,x)\nabla\bigl(\widetilde b\varphi)(s,x)\bigr) \Bigr)\,\dd x\dd s \\
& = \int_t^T\int_{\Rb^d} \varphi(s,x) F^\eps\Bigl( s,x,v^\eps(s,x), \sigma(s,x)^\top\partial_x v^\eps(s,x) \Bigr)\,\dd x\dd s,}
where $\widetilde b:=\frac12\sigma^\top\nabla\sigma+b$. Then $v^\eps$ is a weak solution of the semilinear PDE associated with the Markovian nonlinearity $F^\eps$. Moreover, for the forward process $X$ defined in \eqref{eq-SDE}, the processes
\eqlnostar{eq-pde solution}{
\Yeps_t = v^\eps(t,X_t),
\qquad \Zeps_t = \sigma(t,X_t)^\top\partial_x v^\eps(t,X_t),
}
solve the regularized RBSDE \eqref{eq-ReRBSDE}.
\end{proposition}

In view of this representation, the next theorem shows that the variance term $\Var(\YcZc_0)$ can be controlled by the approximation error of neural networks for the Markovian representation of $\Zeps$, up to a discretization-regularization remainder. In particular, for fixed $\eps$ and sufficiently small $h$, the DBS training objective can be made small whenever the function $\Zeps_t = \sigma(t,X_t)^\top \partial_x v^\eps(t,X_t)$ can be approximated accurately by neural networks.

\begin{theorem}\label{thm:controlofvar}
Assume the same assumptions as in \autoref{thm:deep backward scheme approx error}. Let $\gamma_0, \gamma_1,\gamma_2>0$ be chosen such that $1-40\gamma_1 C_f^2>0$ and $0<C_6<1$,  where
\eqlnostar{eq-C5}{& C_5 := \frac{2}{\gamma_2} + 64\gamma_2 \left(C_f^2+C_{\phi}^2\eps^{-2}\overline{\kappa}_{\infty}^2\right),\\
\label{eq-C6}
& C_6:= 1- 2\exp(C_5 T)\Big(1+\frac{2h}{\gamma_2}\Big)C_4.}
Then, we have, for sufficiently small $h\le\frac{\gamma_2}{2}$,
\eqlnostar{}{\Var\left(\YcZc_{0}\right) \le &\; 2C_6^{-1}\exp(C_5 T)\Big(1+\frac{2h}{\gamma_2}\Big)\sum_{i=0}^{N-1}h\Eb\left|\Zc_i(X_{t_i}; \theta) - \sigma({t_i},X_{t_i})^\top \partial_x v^\eps({t_i},X_{t_i})\right|^2\\
& + C\exp(C_1 T + C_3 T + C_5 T)\eps^{-4}h,}
where $v^\eps(t,x)$ is given in \autoref{propositon:nonlinear Feynman-Kac formula}, and $C_1, C_3, C_4$ are defined in \eqref{eq-C1}, \eqref{eq-C3}, \eqref{eq-C4}, respectively.
\end{theorem}

\begin{proof}
Recall from \eqref{eq-DBS Y l2 error intermediate 1} that 
\eqlnostar{}{\Eb\left|\Delta \YcbZc_{i+1}\right|^2 & + \Eb\int_{t_i}^{t_{i+1}}\left|\Zc_i - \Zct_s\right|^2\dd s \\
= & \; \Eb\left|\Delta \YcbZc_{i} - \feps\left(X_{i+1},\YcZc_{i+1},\Zc_i\right)h + \Eb_i\left[\feps\left(X_{i+1},\YcZc_{i+1},\Zc_i\right)h\right]\right|^2\\
= &\; \Eb\left|\Delta \YcbZc_{i}\right|^2 + \Eb\left|\feps\left(X_{i+1},\YcZc_{i+1},\Zc_i\right)h - \Eb_i\left[\feps\left(X_{i+1},\YcZc_{i+1},\Zc_i\right)h\right]\right|^2 \\
& - 2\Eb\left[\Delta \YcbZc_{i}\left(\feps\left(X_{i+1},\YcZc_{i+1},\Zc_i\right)h - \Eb_i\left[\feps\left(X_{i+1},\YcZc_{i+1},\Zc_i\right)h\right]\right)\right]\\
\ge &\; \Big(1-\frac{h}{\gamma_2}\Big)\Eb\left|\Delta \YcbZc_{i}\right|^2 - \frac{\gamma_2}{h}\Eb\left|\feps\left(X_{i+1},\YcZc_{i+1},\Zc_i\right)h - \Eb_i\left[\feps\left(X_{i+1},\YcZc_{i+1},\Zc_i\right)h\right]\right|^2.
}
Hence, for sufficiently small $h$ such that $\frac{2h}{\gamma_2}\le 1$,
\eqlnostar{}{\Eb\left|\Delta \YcbZc_{i}\right|^2\le &\; \Big(1+\frac{2h}{\gamma_2}\Big)\Eb\left|\Delta \YcbZc_{i+1}\right|^2 + \Big(1+\frac{2h}{\gamma_2}\Big)\Eb\int_{t_i}^{t_{i+1}}\left|\Zc_i - \Zct_s\right|^2\dd s \\
& + \Big(1+\frac{2h}{\gamma_2}\Big)\gamma_2 h\Eb\left|\feps\left(X_{i+1},\YcZc_{i+1},\Zc_i\right) - \Eb_i\left[\feps\left(X_{i+1},\YcZc_{i+1},\Zc_i\right)\right]\right|^2 \\
\le &\; \Big(1+\frac{2h}{\gamma_2}\Big)\Eb\left|\Delta \YcbZc_{i+1}\right|^2 + \Big(1+\frac{2h}{\gamma_2}\Big)\Eb\int_{t_i}^{t_{i+1}}\left|\Zc_i - \Zct_s\right|^2\dd s \\
& + 2\gamma_2 h\Eb\Big|\feps\left(X_{i+1},\YcZc_{i+1},\Zc_i\right) - \feps\left(X_{i},\Eb_i\left[\YcZc_{i+1}\right],\Zc_i\right)\\
& + \feps\left(X_{i},\Eb_i\left[\YcZc_{i+1}\right],\Zc_i\right) - \Eb_i\left[\feps\left(X_{i+1},\YcZc_{i+1},\Zc_i\right)\right]\Big|^2\\
\le &\; \Big(1+\frac{2h}{\gamma_2}\Big)\Eb\left|\Delta \YcbZc_{i+1}\right|^2 + \Big(1+\frac{2h}{\gamma_2}\Big)\Eb\int_{t_i}^{t_{i+1}}\left|\Zc_i - \Zct_s\right|^2\dd s\\
& + 32\gamma_2 \left(C_f^2+C_{\phi}^2\eps^{-2}\overline{\kappa}_{\infty}^2\right)h\left(Ch + \Eb\left|\YcZc_{i+1} - \Eb_i\left[\YcZc_{i+1}\right]\right|^2\right) \\
\le &\; \Big(1+\frac{2h}{\gamma_2}\Big)\Eb\left|\Delta \YcbZc_{i+1}\right|^2 + \Big(1+\frac{2h}{\gamma_2}\Big)\Eb\int_{t_i}^{t_{i+1}}\left|\Zc_i - \Zct_s\right|^2\dd s\\
& + 32\gamma_2 \left(C_f^2+C_{\phi}^2\eps^{-2}\overline{\kappa}_{\infty}^2\right)h\Big(Ch + 2\Eb\int_{t_i}^{t_{i+1}}\left|\Zct_s\right|^2\dd s + 2\Eb\left|\Delta \YcbZc_{i+1}\right|^2\Big)\\
\le &\; \left(1+C_5 h\right)\Eb\left|\Delta \YcbZc_{i+1}\right|^2 + \Big(1+\frac{2h}{\gamma_2}\Big)\Eb\int_{t_i}^{t_{i+1}}\left|\Zc_i - \Zct_s\right|^2\dd s\\
& + C\gamma_2 h\left(1+\eps^{-2}\right)\Big(h + \Eb\int_{t_i}^{t_{i+1}}\left|\Zct_s\right|^2\dd s\Big),}
where 
$C_5 := \frac{2}{\gamma_2} + 64\gamma_2 \left(C_f^2+C_{\phi}^2\eps^{-2}\overline{\kappa}_{\infty}^2\right)$,
and the third inequality is obtained in the same way as \eqref{eq-DBS approx error intermediate 7}, the fourth follows from \eqref{eq-martingale representation theorem of ycbzc org}.

Applying Gronwall's lemma, we have
\eqlnostar{}{\underset{0\leq i\leq N-1}{\max}\,\Eb\left|\Delta \YcbZc_{i}\right|^2\le &\;\exp(C_5 T)\Eb\left|\Delta \YcbZc_{N}\right|^2 + \exp(C_5 T)\Big(1+\frac{2h}{\gamma_2}\Big)\sum_{i=0}^{N-1}\Eb\int_{t_i}^{t_{i+1}}\left|\Zc_i - \Zct_s\right|^2\dd s \\
& + C\exp(C_5 T)\Big(h + \eps^{-2}h\Big)\Big(1 + \Eb\int_{0}^{T}\left|\Zct_s\right|^2\dd s\Big)\\
\le &\; \exp(C_5 T)\Big(1+\frac{2h}{\gamma_2}\Big)\sum_{i=0}^{N-1}\Eb\int_{t_i}^{t_{i+1}}\left|\Zc_i - \Zct_s\right|^2\dd s + C\exp(C_5 T)\left(h + \eps^{-2}h\right).}
Therefore, we obtain
\eqlnostar{}{\Var\left( \YcZc_{0}\right) \le &\; \exp(C_5 T)\Big(1+\frac{2h}{\gamma_2}\Big)\sum_{i=0}^{N-1}\Eb\int_{t_i}^{t_{i+1}}\left|\Zc_i - \Zct_s\right|^2\dd s + C\exp(C_5 T)\left(h + \eps^{-2}h\right)\\
\le &\; 2\exp(C_5 T)\Big(1+\frac{2h}{\gamma_2}\Big)\sum_{i=0}^{N-1}\Eb\int_{t_i}^{t_{i+1}}\left|\Zc_i - \Zeps_s\right|^2\dd s \\
\label{eq-optimization error z estimate intermediate 1}
& + 2\exp(C_5 T)\Big(1+\frac{2h}{\gamma_2}\Big)\sum_{i=0}^{N-1}\Eb\int_{t_i}^{t_{i+1}}\left|\Zeps_s - \Zct_s\right|^2\dd s + C\exp(C_5 T)\left(h + \eps^{-2}h\right).}
Substituting \eqref{eq-DBS approx error intermediate 6} into \eqref{eq-optimization error z estimate intermediate 1} yields
\eqlnostar{}{\Var\left( \YcZc_{0}\right)
\le &\; 2\exp(C_5 T)\Big(1+\frac{2h}{\gamma_2}\Big)\bigg(\sum_{i=0}^{N-1}\Eb\int_{t_i}^{t_{i+1}}\left|\Zc_i - \Zeps_i\right|^2\dd s + \sum_{i=0}^{N-1}\Eb\int_{t_i}^{t_{i+1}}\left|\Zeps_i - \Zeps_s\right|^2\dd s\bigg)\\
& + 2\exp(C_5 T)\Big(1+\frac{2h}{\gamma_2}\Big)C_4\Var\left(\YcZc_{0}\right) + C\exp(C_1 T + C_3 T + C_5 T)\left(h + \eps^{-2}h\right)\\
& + C2\exp(C_5 T)\Big(1+\frac{2h}{\gamma_2}\Big)\left(\eps^{-2}h + h + \eps^{-4}h\right).}
Now define
$C_6:= 1- 2\exp(C_5 T)\left(1+\frac{2h}{\gamma_2}\right)C_4$.
Rearranging the terms and using  $0 < C_6 < 1$ gives
\eqlnostar{}{\Var\left( \YcZc_{0}\right)
\le &\; 2C_6^{-1}\exp(C_5 T)\Big(1+\frac{2h}{\gamma_2}\Big)\bigg(\sum_{i=0}^{N-1}\Eb\int_{t_i}^{t_{i+1}}\left|\Zc_i - \Zeps_i\right|^2\dd s + \sum_{i=0}^{N-1}\Eb\int_{t_i}^{t_{i+1}}\left|\Zeps_i - \Zeps_s\right|^2\dd s\bigg)\\
& + CC_6^{-1}\left(h + \eps^{-2}h\right) + C\exp(C_1 T + C_3 T + C_5 T)\left(\eps^{-2}h + h + \eps^{-4}h\right)\\
\le &\; 2C_6^{-1}\exp(C_5 T)\Big(1+\frac{2h}{\gamma_2}\Big)\bigg(\sum_{i=0}^{N-1}\Eb\int_{t_i}^{t_{i+1}}\left|\Zc_i - \Zeps_i\right|^2\dd s + C\eps^{-2}h\bigg) \\
& + C\exp(C_1 T + C_3 T + C_5 T)\left(\eps^{-2}h + h + \eps^{-4}h\right)\\
\le &\; 2C_6^{-1}\exp(C_5 T)\Big(1+\frac{2h}{\gamma_2}\Big)\sum_{i=0}^{N-1}h\Eb\left|\Zc_i - \Zeps_i\right|^2\\
& + C\exp(C_1 T + C_3 T + C_5 T)\left(\eps^{-2}h + h + \eps^{-4}h\right),}
where the second inequality follows from the path regularity for $\Zeps$ in \citet[Theorem 17]{agarwal2026numerical}. Finally, combining with \eqref{eq-pde solution}, we complete the proof.
\end{proof}

This result provides a theoretical justification for the DBS training objective. Indeed, it shows that, for fixed $\eps$, the loss $\Var(\YcZc_0)$ is controlled by the approximation error of the network class for the target map $\sigma^\top \partial_x v^\eps$, together with a discretization-regularization remainder. Combined with the error estimate in \autoref{thm:true DBS approx error}, this yields an error decomposition for approximating the original RBSDE. Since the constants are not uniform in $\eps$, this should be interpreted as a fixed-$\eps$ error estimate together with the separate regularization error, rather than as a uniform convergence rate as $\eps\downarrow0$.

\begin{remark}
The neural-network approximation term in the above theorem is not pointwise in $x$, but measured in $L^2$ under the law of the forward process $X_{t_i}$. Thus, for each fixed time step $t_i$, it is sufficient to approximate the target map $x \mapsto \sigma(t_i, x) \partial_x v^\eps(t_i, x)$ on the region of the state space effectively explored by the forward process $X_{t_i}$. Standard feedforward neural networks are therefore a natural approximation class, since classical universal-approximation results \citep{hornik1989multilayer,hornik1990universal,leshno1993multilayer} establish density in $L^p(\mu)$ for finite measures $\mu$.

We emphasize, however, that the target map itself depends on the regularization parameter $\eps$. As $\eps\downarrow 0$, the regularized generator becomes increasingly steep near the obstacle, since derivatives of $\phieps$ scale as powers of $\eps^{-1}$. Consequently, the functions $\sigma(t_i, x) \partial_x v^\eps(t_i, x)$ may become more difficult to approximate for small $\eps$, especially near the contact region. To obtain convergence as $\eps\downarrow0$, one therefore needs the neural-network classes to approximate the family $\{\sigma(t_i, x) \partial_x v^\eps(t_i, x)\}_{\eps>0}$ with sufficient accuracy as $\eps$ decreases, possibly by increasing the network size or other complexity parameters. Establishing such uniform-in-$\eps$ approximation properties would require additional regularity or structural information on the family $v^\eps$.

Finally,  the above universal-approximation results are largely existential in nature: they do not specify how to construct or train the approximating networks. Furthermore, quantitative approximation or convergence rates generally require additional assumptions on the target function class or on the underlying problem structure; see \cite{yarotsky2017error,lu2021deep,de2021approximation,montanelli2020error,schmidt2020nonparametric} for representative results.
\end{remark}

\section{Numerical Experiments on the Pricing of American-type Options}\label{sec:numerical results}

In this section, we present numerical experiments on the pricing of American-type options to assess the accuracy of the proposed deep learning-based solvers for RBSDEs in high-dimensional settings. 

Denote by $r$ the risk-free interest rate. We consider a $d$-dimensional underlying asset, whose price process $X=(X^1,\ldots,X^d)^\top$ evolves, under the risk-neutral measure, according to a geometric Brownian motion
\eqlnostar{}{
\frac{\dd X_t^i}{X_t^i} = r\,\dd t + \sigma_i\,\dd W_t^i,
\qquad i=1,\ldots,d,
}
where $W:=(W^1,\ldots,W^d)^\top$ is a $d$-dimensional Brownian motion with instantaneous correlations:
$ \dd\langle W^i,W^j\rangle_t = \rho_{ij}\;\dd t$. Hereafter, we use $\rho$ to denote the correlation matrix of $W$. 

We consider the pricing of an American option written on the underlying asset process $X$, with payoff function $\Phi(\cdot)$. The corresponding value process is given by $$
Y_t := \sup_{\tau \in \mathcal{T}_{t,T}} \Eb[e^{-r(T-\tau)}\Phi(X_\tau)\vert \mathcal{F}_t],
$$ 
where $\mathcal{T}_{t,T}$ denotes the set of all stopping times in $[t,T]$. It is well known that $Y_t$, together with the adjoint process $Z_t$, are characterized by the RBSDE \eqref{eq-RBSDE}, with obstacle process $S_t = \Phi(X_t)$, and generator $f(t,x,y,z)=-ry$. In the following examples, we specify the functional form of $\Phi$ for different types of call and put options.

For all numerical experiments, we fix the parameters
\eqlnostar{}{
r = 0.02,\quad \sigma = 0.2,\quad T = 1,\quad K^{\mathrm{str}} = 100, \quad X_0^i\equiv 100, \quad \rho_{i,j}\equiv 0, \quad N = 64,\quad \eps = h = \frac{T}{N}.
}
Regarding the hyperparameters and implementation details, we employ a feedforward neural network architecture consisting of three hidden layers, each with $d+10$ neurons, where $d$ denotes the dimension of the state vector $X$. The ReLU activation function is used throughout. Our implementation is based on the \texttt{PyTorch} library. We also incorporate batch normalization the form \texttt{BatchNorm1d} $(d+1;\,\epsilon=10^{-6},\,\text{momentum}=0.99)$, applied to the input layer. For the DBS, the network is trained using a batch size of 512 simulated paths for 3000 training iterations. Optimization is performed using the Adam algorithm (\texttt{torch.optim.Adam}) with an initial learning rate of \(5\times10^{-3}\), combined with a stepwise learning-rate scheduler (\texttt{StepLR}) with step size 100 and decay factor 0.99. The DBS price is obtained from Step 4 of the algorithm: after training the control network $\bm{\Zc}$, we simulate $10^6$ independent backward paths using the learned network and evaluate their mean as a proxy for $\Eb[\YcZc_0]$. For the DFS, we use the same network architecture and batch size, but train the model for 5000 iterations with an initial learning rate of $5\times10^{-2}$ and the same scheduler; the reported DFS price is given directly by the trained parameter $y$. All numerical experiments are conducted on a device equipped with an AMD Ryzen 9 9950X 16-core processor, 96 GB of RAM, and an NVIDIA GeForce RTX 5090 GPU (32 GB memory).

Before presenting the numerical examples, we note that the decomposition \eqref{eq-representation of S} is used to construct the regularized BSDE. In particular, the regularized generator \eqref{eq-regularized generator} depends on the processes $S$, $U$, and $V$, which must therefore be identified for each payoff specification.

\subsection{Geometric Put}

Let us consider the pricing of American geometric put options written on a basket of assets. Define
\eqlnostar{}{
\overline X_t^g := \bigg(\prod_{i=1}^d X_t^i\bigg)^{1/d},\quad M = \left(\frac{\sigma_1}{d},\frac{\sigma_2}{d},\cdots,\frac{\sigma_d}{d}\right)^\top.
}
The payoff function (obstacle) is given by $\Phi(X_t) = S_t = \big(K^{\text{str}} - \overline X_t^g\big)^+$, where $K^{\mathrm{str}}$ denotes the strike price. Then $\overline X$ follows the following dynamics
\eqlnostar{}{
\frac{\dd \overline X_t^g}{\overline X_t^g} = \bigg(r - \frac{1}{2d}\sum_{i=1}^d \sigma_i + \frac{1}{2}M^\top \rho M\bigg)\,\dd t + M \cdot\dd W_t,
\qquad
\overline X_0 = \bigg(\prod_{i=1}^d X_0^i\bigg)^{1/d}.
}
By the It\^o-Tanaka formula, we know that the obstacle process admits the decomposition
\eqlnostar{}{
\dd S_t
= -\Ib_{\{\overline X_t^g\le K^{\mathrm{str}}\}}\overline X_t^g\Big(r - \frac{1}{2d}\sum_{i=1}^d \sigma_i + \frac{1}{2}M^\top \rho M\Big)\,\dd t
-\Ib_{\{\overline X_t^g\le K^{\mathrm{str}}\}}\overline X_t^g\,M \cdot \dd W_t
+ \frac{1}{2}\,\dd L_t^{K^{\mathrm{str}}}(\overline X^g),
}
where $L^{K^{\mathrm{str}}}(\overline X^g)$ denotes the local time process of $\overline X^g$ at level $K^{\mathrm{str}}$. Therefore, we can obtain the representation for $U$ and $V$ as
\eqlnostar{}{U_t = -\Ib_{\{\overline X_t^g\le K^{\mathrm{str}}\}}\overline X_t^g\Big(r - \frac{1}{2d}\sum_{i=1}^d \sigma_i + \frac{1}{2}M^\top \rho M\Big),\qquad V_t = -\Ib_{\{\overline X_t^g\le K^{\mathrm{str}}\}}\overline X_t^g\,M.}
Since this problem can be reduced to the pricing of a one-dimensional American put option written on the geometric average, it allows us to compute a highly accurate reference solution. In our numerical experiments, we use a binomial tree method to approximate the one-dimensional American option price written on the geometric average, with $10^4$ time steps, and treat this solution as a benchmark for assessing the accuracy of the proposed deep learning-based solver. The results for the DFS and the DBS are reported in Table \ref{tab:american-geom-put}. Both schemes attain consistently high accuracy across all tested dimensions, with small relative errors and root mean squared errors (RMSE). 
\begin{table}[h!]
\centering
\begin{tabular}{ccccccccc}
\hline
& & \multicolumn{3}{c}{DFS} & & \multicolumn{3}{c}{DBS} \\
\cline{3-5} \cline{7-9}
$d$ & Ref.\ value
& Mean (Std) & Rel.\ err (\%) & RMSE
&
& Mean (Std) & Rel.\ err (\%) & RMSE \\
\hline
1   & 7.1107 & 7.1431 (0.0050) & 0.4562 & 0.0328 & & 7.1380 (0.0040) & 0.3835 & 0.0276 \\
5   & 3.3518 & 3.3678 (0.0018) & 0.4790 & 0.0161 & & 3.3613 (0.0007) & 0.2859 & 0.0096 \\
10  & 2.4014 & 2.4115 (0.0017) & 0.4207 & 0.0102 & & 2.4080 (0.0005) & 0.2750 & 0.0066 \\
20  & 1.7143 & 1.7213 (0.0017) & 0.4124 & 0.0073 & & 1.7191 (0.0003) & 0.2814 & 0.0048 \\
30  & 1.4057 & 1.4111 (0.0016) & 0.3842 & 0.0056 & & 1.4097 (0.0002) & 0.2878 & 0.0041 \\
40  & 1.2205 & 1.2254 (0.0018) & 0.4057 & 0.0053 & & 1.2241 (0.0002) & 0.2986 & 0.0037 \\
50  & 1.0935 & 1.0980 (0.0020) & 0.4067 & 0.0048 & & 1.0969 (0.0003) & 0.3091 & 0.0034 \\
100 & 0.7766 & 0.7803 (0.0022) & 0.4751 & 0.0029 & & 0.7794 (0.0002) & 0.3538 & 0.0028 \\
\hline
\end{tabular}
\caption{Numerical results for American geometric put options using DFS and DBS. Mean estimates are reported with standard deviations in parentheses, based on 50 independent runs.}
\label{tab:american-geom-put}
\end{table}

\subsection{Basket Call}
As for American basket call options written on a basket of assets, define
\eqlnostar{}{
\overline X_t^a := \frac{1}{d}\sum_{i=1}^d X_t^i, \qquad 
a_t := \bigg(\frac{\sigma_1 X_t^1}{d},\frac{\sigma_2 X_t^2}{d},\cdots,\frac{\sigma_d X_t^d}{d}\bigg)^\top.
}
The payoff function is given by 
$\Phi(X_t) = S_t = \big(\overline X_t^a - K^{\mathrm{str}}\big)^+$. Then $\overline X^a$ follows the dynamics
\eqlnostar{}{
\dd \overline X_t^a = r\,\overline X_t^a\,\dd t + a_t \cdot \dd W_t,
\qquad 
\overline X_0^a = \frac{1}{d}\sum_{i=1}^d X_0^i.
}
By the It\^o-Tanaka formula, we know that the obstacle process admits the decomposition
\eqlnostar{}{
\dd S_t
= \Ib_{\{\overline X_t^a > K^{\mathrm{str}}\}}\, r\,\overline X_t^a\,\dd t
+ \Ib_{\{\overline X_t^a > K^{\mathrm{str}}\}}\, a_t \cdot \dd W_t
+ \frac{1}{2}\,\dd L_t^{K^{\mathrm{str}}}(\overline X^a).
}
Therefore, we obtain the representation for $U$ and $V$ as
\eqlnostar{}{
U_t = \Ib_{\{\overline X_t^a > K^{\mathrm{str}}\}}\, r\,\overline X_t^a,
\qquad 
V_t = \Ib_{\{\overline X_t^a > K^{\mathrm{str}}\}}\, a_t.
}
For call options with a nonnegative risk-free interest rate $r \ge 0$ and zero dividends, it is well known that early exercise is not optimal. This aligns with our parameter specifications. Therefore, the price of the American option coincides with that of the corresponding European option under these constraints. In our numerical experiments, we compute the Monte Carlo (MC) results with $10^6$ sample paths and employ antithetic variates for the corresponding European option. These MC values are subject to statistical error and should not be regarded as exact. The results for the DFS and the DBS are reported in Table \ref{tab:american-basket-call}. Again, both schemes show consistently high accuracy across all tested dimensions. The relative differences are very small, and the results from the DFS and the DBS are highly consistent with each other and closely aligned with the MC estimates.
\begin{table}[h!]
\centering
\begin{tabular}{cccccccc}
\hline
& & \multicolumn{2}{c}{DFS} & & \multicolumn{2}{c}{DBS} \\
\cline{3-4} \cline{6-7}
$d$ & MC Mean (Std)
& Mean (Std) & Rel.\ diff (\%)
&
& Mean (Std) & Rel.\ diff (\%) \\
\hline
5   & 4.6364 (0.0048) & 4.6359 (0.0031) & 0.0099 & & 4.6358 (0.0006) & 0.0123 \\
10  & 3.6321 (0.0032) & 3.6321 (0.0034) & 0.0007 & & 3.6319 (0.0006) & 0.0055 \\
20  & 2.9450 (0.0018) & 2.9440 (0.0035) & 0.0362 & & 2.9443 (0.0005) & 0.0239 \\
50  & 2.3841 (0.0010) & 2.3834 (0.0026) & 0.0269 & & 2.3833 (0.0003) & 0.0312 \\
100 & 2.1498 (0.0006) & 2.1498 (0.0029) & 0.0008 & & 2.1493 (0.0004) & 0.0240 \\
200 & 2.0320 (0.0004) & 2.0321 (0.0042) & 0.0080 & & 2.0316 (0.0003) & 0.0192 \\
\hline
\end{tabular}
\caption{Numerical results for American basket call options using DFS and DBS, based on 50 independent runs. Relative differences are computed with respect to European basket call prices obtained via Monte Carlo (MC) estimation.}
\label{tab:american-basket-call}
\end{table}

\subsection{Max Call}
As for American max call options written on a basket of assets, the payoff function is given by 
$\Phi(X_t) = S_t = \big(\max_{1\le i \le d} X_t^i - K^{\mathrm{str}}\big)^+$. At each time $t$, let us reorder the components of $X_t=(X_t^1,\ldots,X_t^d)^\top$ such that $X_t^{(1)} \le \cdots \le X_t^{(d)}$, so that $X_t^{(d)} := \max_{1\le i \le d} X_t^i$, while $X_t^{(d-1)}$ denotes the second-largest component. Then, by applying the decomposition formula of ranked continuous semimartingales (see \citet[Corollary 2.6]{banner2008local}), we have
\eqlnostar{}{
\dd X_t^{(d)} & = \sum_{i=1}^{d}\Ib_{\big\{X_t^i=X_t^{(d)}\big\}}\dd X_t^i + \frac{1}{2}\dd L_t^{0}\left(X^{(d)-(d-1)}\right) = r X_t^{(d)}\dd t
+ b_t \cdot \dd W_t
+ \frac{1}{2}\dd L_t^{0}\left(X^{(d)-(d-1)}\right),
}
where
$b_t := \big(\sigma_1 X_t^1 \Ib_{\{X_t^1 = X_t^{(d)}\}},\,
\sigma_2 X_t^2 \Ib_{\{X_t^2 = X_t^{(d)}\}},\,
\cdots,\,
\sigma_d X_t^d \Ib_{\{X_t^d = X_t^{(d)}\}}\big)^\top$ and $X^{(d)-(d-1)}:=X^{(d)} - X^{(d-1)}$.

By the It\^o-Tanaka formula, the obstacle process admits the decomposition
\eqlnostar{}{
\dd S_t
= \Ib_{\{X_t^{(d)} > K^{\mathrm{str}}\}} r X_t^{(d)}\dd t
+ \Ib_{\{X_t^{(d)} > K^{\mathrm{str}}\}} b_t \cdot \dd W_t
+ \frac{1}{2}\Ib_{\{X_t^{(d)} > K^{\mathrm{str}}\}}\dd L_t^{0}\big(X^{(d)-(d-1)}\big)
+ \frac{1}{2}\dd L_t^{K^{\mathrm{str}}}\big(X_t^{(d)}\big).
}
Therefore, we obtain the representation for $U$ and $V$ in this case as
\eqlnostar{}{
U_t = \Ib_{\{X_t^{(d)} > K^{\mathrm{str}}\}} r X_t^{(d)},
\qquad 
V_t = \Ib_{\{X_t^{(d)} > K^{\mathrm{str}}\}} b_t.
}
Regarding the numerical experiments, we report Monte Carlo (MC) results for the corresponding European options, based on $10^6$ sample paths with antithetic variates, in Table \ref{tab:american-max-call}, together with the results for the DBS and the DFS. Compared with the basket call case, the relative differences are slightly larger relative to the MC estimates. Nevertheless, both the DFS and the DBS remain close to each other and exhibit good agreement with the MC estimates across all dimensions.
\begin{table}[h!]
\centering
\begin{tabular}{cccccccc}
\hline
& & \multicolumn{2}{c}{DFS} & & \multicolumn{2}{c}{DBS} \\
\cline{3-4} \cline{6-7}
$d$ & MC Mean (Std)
& Mean (Std) & Rel.\ diff (\%)
&
& Mean (Std) & Rel.\ diff (\%) \\
\hline
5   & 26.9242 (0.0146) & 26.8971 (0.0328) & 0.1005 & & 26.8954 (0.0029) & 0.1069 \\
10  & 36.2590 (0.0166) & 36.1986 (0.0446) & 0.1664 & & 36.1986 (0.0068) & 0.1665 \\
20  & 45.1883 (0.0156) & 45.0815 (0.0437) & 0.2363 & & 45.0740 (0.0127) & 0.2529 \\
50  & 56.3525 (0.0153) & 56.1656 (0.0330) & 0.3316 & & 56.1715 (0.0130) & 0.3212 \\
100 & 64.4429 (0.0143) & 64.2023 (0.0307) & 0.3733 & & 64.1931 (0.0176) & 0.3876 \\
200 & 72.2992 (0.0141) & 71.9707 (0.0308) & 0.4544 & & 71.9722 (0.0116) & 0.4524 \\
\hline
\end{tabular}
\caption{Numerical results for American max call options using DFS and DBS, based on 50 independent runs. Relative differences are computed with respect to European max call prices obtained via Monte Carlo (MC) estimation.}
\label{tab:american-max-call}
\end{table}

\section{Conclusion}\label{sec:conclusion}
In this paper, we develop deep learning-based solvers for RBSDEs by combining a regularization framework with deep BSDE methods. In particular, we consider both a deep forward scheme (DFS) and a deep backward scheme (DBS) inspired by \citet{han2018solving,wang2018deep} for the regularized problem, and provide a detailed analysis of the DBS. We establish approximation error bounds in terms of the associated loss function, and show that the loss can be made arbitrarily small under the universal approximation property of neural networks. In addition, we present numerical experiments on high-dimensional American-style option pricing problems, which demonstrate the accuracy and effectiveness of the proposed methods.

\singlespacing
\normalem
\printbibliography
\end{document}